\documentclass[11pt,pra,amsmath,amssymb,aps,color,superscriptaddress,notitlepage]{revtex4-1}
\usepackage[utf8]{inputenc}
\usepackage{amsmath,amsfonts,amssymb,amsthm}
\usepackage{graphicx}
\usepackage{bbm,bbold}
\usepackage{tikz}

\usepackage[top=0.8in, bottom=1.5in, left=1.3in, right=1.3in]{geometry}

\definecolor{jens}{rgb}{0,0,0}
\definecolor{all}{rgb}{0,0,0}
\newcommand{\je}{\textcolor{jens}}

\newcommand{\all}{\textcolor{all}}
\newcommand{\jg}[1]{{\color{black} #1}}
\newcommand{\ja}[1]{{\color{black} #1}}
\newcommand{\ec}[1]{{\textcolor{black}{ #1}}}

\newtheorem{thm}{Theorem}
\newtheorem{lm}{Lemma}
\newtheorem{df}[lm]{Definition}
\newtheorem*{lm*}{Lemma}
\newtheorem{co}[thm]{Corollary}
\newtheorem*{co*}{Corollary}

\newcommand{\Trace}[1]{\operatorname{Tr}\left[#1\right]}
\newcommand{\SmallTrace}[1]{\operatorname{Tr}[#1]}
\DeclareMathOperator{\tr}{tr}
\newcommand{\cc}{\mathbbm{C}}
\newcommand{\Ra}{\Rightarrow}
\newcommand{\ra}{\rightarrow}
\newcommand{\tel}[1]{\frac{1}{#1}}\newcommand{\id}{\mathbb{1}}

\newcommand{\norm}[2]{\| #1 \|_{#2}}
\newcommand{\oneNorm}[1]{\norm{#1}{1}}
\newcommand{\infNorm}[1]{\norm{#1}{\infty}}

\newcommand{\oneoneNorm}[2]{\frac{\oneNorm{#1\kl{#2}}}{\oneNorm{#2}}}
\newcommand{\oneToOne}[1]{\norm{#1}{1\ra 1}}
\newcommand{\toNorm}[2]{\norm{#1}{#2\ra #2}}
\newcommand{\sigmafrac}[2]{\frac{\oneNorm{#1\kl{\sigma_{#2}}}}{\oneNorm{\sigma_{#2}}}}
\newcommand{\normfrac}[2]{\frac{\oneNorm{#1\kl{#2}}}{\oneNorm{#2}}}

\newcommand{\pophim}[1]{p_{#1}(\phi^-)}

\newcommand{\eee}{\epsilon}
\newcommand{\x}{\xi}
\newcommand{\tec}[1]{\tau\kl{#1}}
\newcommand{\inv}[1]{\kl{#1}^{-1}}
\newcommand{\idd}{\id}  
\newcommand{\adj}{^{\dagger}}
\newcommand{\maxs}[1]{\max_{\sigma_{#1}}}
\newcommand{\maxt}[1]{\max_{\Trace{\sigma_{#1}}=0}}
\newcommand{\with}{\quad \text{with} \quad}

\newcommand{\MS}{\mathcal{S}}
\newcommand{\PP}{\mathcal{P}}

\newcommand{\Finf}[1]{{F^{\infty}_{#1}}}

\newcommand{\FFinf}[1]{{\mathcal{F}^{\infty}_{#1}}}

\newcommand{\ZZ}{\mathcal{Z}}
\newcommand{\ZZZ}[1]{\mathcal{Z}\kl{#1}}
\newcommand{\EE}{\mathcal{E}}

\newcommand{\AAn}{\mathcal{A}_n}
\newcommand{\AAnn}{\mathcal{A}_{n+2}}
\newcommand{\EEn}{\mathcal{E}_n}
\newcommand{\BBn}{\mathcal{B}_n}
\newcommand{\CCn}{\mathcal{C}_n}
\newcommand{\DDn}{\mathcal{D}_n}
\newcommand{\FF}{\mathcal{F}}

\newcommand{\kl}[1]{\left(#1\right)}
\newcommand{\skl}[1]{\left\{#1\right\}}
\newcommand{\bkl}[1]{\left| #1 \right|}

\newcommand{\bra}[1]{\left\langle #1 \right|}
\newcommand{\ket}[1]{\left| #1\right\rangle}
\newcommand{\braket}[2]{\left\langle #1 \right|\left. #2 \right\rangle}

\newcommand{\Braket}[3]{\left\langle #1 \right|\left. #2 \right|\left. #3\right\rangle}

\usepackage{times}

\begin{document}
%\listofchanges\newpage%%

\title{Renormalising entanglement distillation}
\author{Stephan Waeldchen}
\affiliation{Dahlem Center for Complex Quantum Systems, Freie Universit{\"a}t Berlin, 14195 Berlin, Germany}
\author{Janina Gertis}
\affiliation{Dahlem Center for Complex Quantum Systems, Freie Universit{\"a}t Berlin, 14195 Berlin, Germany}
\author{Earl T.\ Campbell}
\affiliation{Dahlem Center for Complex Quantum Systems, Freie Universit{\"a}t Berlin, 14195 Berlin, Germany}
\affiliation{Department of Physics and Astronomy,
University of Sheffield, Sheffield S3 7RH, UK}
\author{Jens Eisert}
\affiliation{Dahlem Center for Complex Quantum Systems, Freie Universit{\"a}t Berlin, 14195 Berlin, Germany}

\date{\today}

\begin{abstract}
Entanglement distillation refers to the task of transforming a collection of weakly entangled pairs into fewer highly entangled ones. It is a core ingredient in quantum repeater protocols, needed to transmit entanglement over arbitrary distances in order to realise quantum key distribution schemes. Usually, it is assumed that the initial entangled pairs are i.i.d. distributed and uncorrelated with each other, an assumption that might not be reasonable at all in any entanglement generation process involving memory channels.  Here, we introduce a framework that captures entanglement distillation in the presence of natural correlations arising from memory channels. Conceptually, we bring together ideas from condensed-matter physics - that of renormalisation and of matrix-product states and operators - with those of local entanglement manipulation, Markov chain mixing, and quantum error correction. We identify meaningful parameter regions for which we prove convergence to maximally entangled states, arising as the fixed points of a matrix-product operator renormalisation flow.
\end{abstract}
\maketitle
 
It has long been noted in the field of quantum information science that entanglement constitutes the key resource in various information processing and specifically communication tasks \cite{Bennett}. Secure quantum key distribution necessarily relies on entanglement, even in prepare and measure schemes \cite{GisinReview,GisinReview2,Lutkenhaus}.  A central goal in quantum information science has %therefore 
been the development of techniques to transform less \ja{useful} %usable 
forms of entanglement into more suitable ones, and \ja{enhancing} %building 
our understanding of the laws governing the manipulation of entanglement.  The task of entanglement distillation specifically captures the resource character of entanglement, in that it aims at preparing maximally entangled states from noisy or less entangled ones \cite{Bennett}. 
 The concept of {\it distillable entanglement} grasps the  maximum rate at which this is asymptotically possible, starting from a collection of many identically prepared quantum systems -- and is hence of profound interest in the conceptual foundations of the field. 
Distillation steps are part of {\it quantum repeater protocols} \cite{Repeater1,Repeater2,Repeater3}, necessary to distribute entanglement over arbitrary distances \all{using} noisy quantum channels: In such a scheme, entanglement is being established between different repeater stations and transferred to the final designated nodes via suitable entanglement swapping steps. Distillation schemes thought of in this context are often iterative schemes, such as the recurrence protocol \cite{DeutschEkert,Bennett} and deterministic protocols based on error-correction codes~\ec{\cite{Bennett,Hartmann,Liang,Ying13}}. While iterative schemes do not achieve the maximum rates set by the distillable entanglement, they require less sophisticated and \all{more} 
practically feasible operations. 

 The silent assumption in almost exclusively all of the proposed distillation schemes, however, is that the initial resources have been identically prepared and show no correlations whatsoever. 
 \all{While this is surely a good assumption in many preparations, it might not be reasonable at all in others.}
Whenever memory effects or channels \cite{HowLongForget,PlenioMemory,Memory} are involved, one expects some correlations between the involved entangled pairs, going beyond an i.i.d.\ setting. 
These correlations are expected to  decay rapidly over several pairs sent through a channel -- reflecting the natural correlation structure arising from a {\it memory channel} (see Fig.~\ref{fig:mpo}).  
\all{This may even be a desired feature, if resetting the system may be too elaborate or too costly in 
resources of time.} The mathematical definition of distillable entanglement in the presence of correlations has been developed~\cite{CopyCorrelatedSources,CorrelatedSources}. \ec{Yet,} the important practical problem of \ja{distilling} %how to distill 
entanglement from correlated pairs arising from quantum memory channels is still wide open.
 
In this work, we propose a conceptually novel way forward to solve this problem, bringing together ideas from entanglement theory with those of condensed matter theory, specifically of renormalisation \ja{techniques} and tensor network \ja{methods}.% states.
We start by identifying the natural class of states arising from preparations and memory channels as bi-partite {\it matrix-product operators}
(MPO)
\cite{Mixed,Zwolak}. Such classes of states are usually considered in the condensed matter context to capture thermal many-body states or those arising from open systems dynamics \cite{Mixed,Zwolak,Positive,Banuls}. 
Here, we encounter \all{natural bi-partite} instances thereof. 
Entanglement distillation is then identified as a {\it renormalisation of bi-partite matrix-product operators}. %Viewed in this mindset, 
The methods are inspired \all{by} and derive from renormalisation \cite{RenormalisingMPS} of {\it matrix-product states} \cite{FCS,MPSSurvey,Schollwock201196,TensorNetworkReview,OrusReview,SchuchReview}, 
again from many-body theory. 

% EARL: I commented out the text below as I couldn't understand it and we needed to save some space.

% There are yet on the one hand rather profound differences arising both from the very different application, here in entanglement manipulation. On the other, we are concerned with (i) bi-partite multi-partite states as well as (ii) matrix product operators \cite{Zwolak,Mixed} instead of pure states. We will see that this language is most suitable for the problem at hand. 

\begin{figure}[b!]
\includegraphics[width=0.53\textwidth]{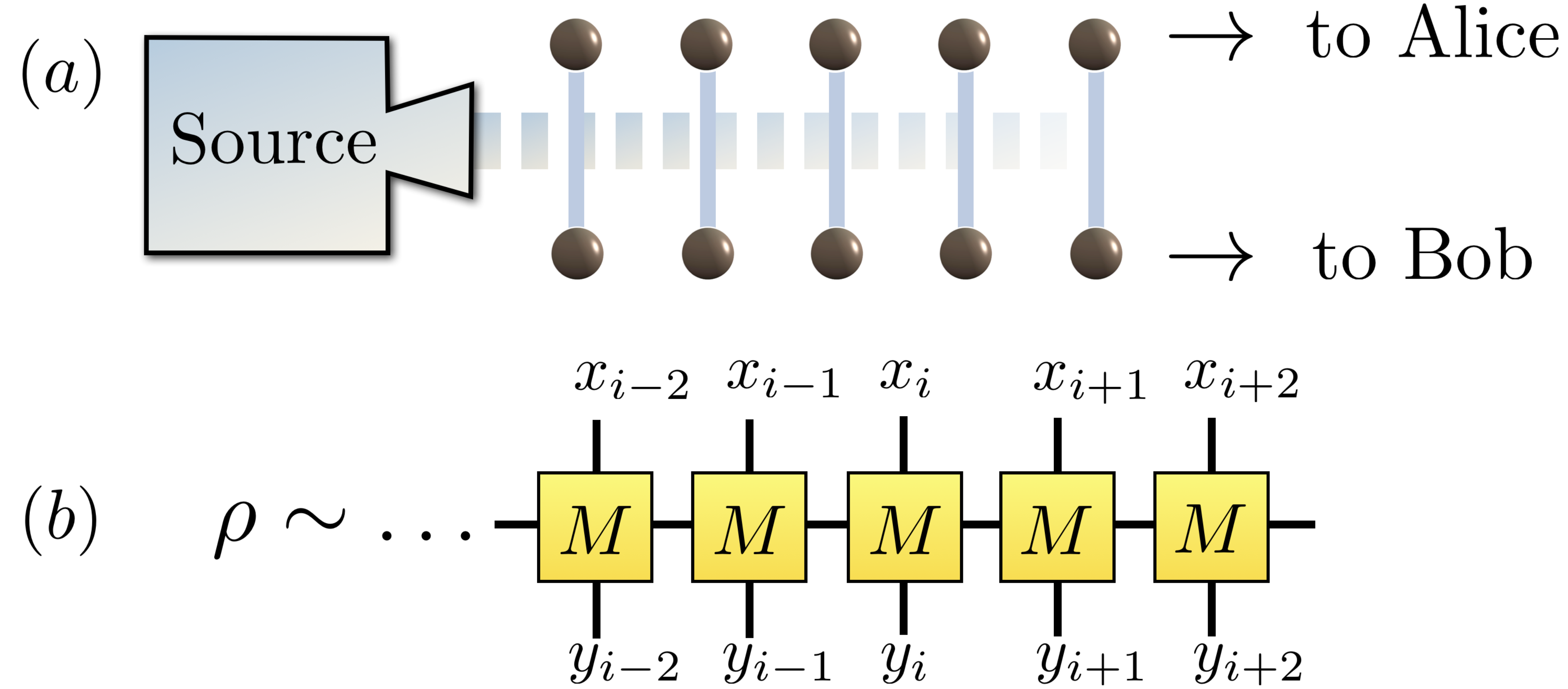}
\caption{(a) Source with memory emitting weakly correlated photon pairs and (b) the MPO $\rho$ naturally describing this setting. 
%In the process of one iteration, $n$ boxes are combined and contribute to one new box resulting in a new MPO.
}
\label{fig:mpo}
\end{figure}

Specifically, we show that both the recurrence protocol and \all{an} error correction \all{based}
protocol \all{\cite{Bennett}}
converge to pairs of maximally entangled pure states for \ec{suitably} correlated pairs which are naturally described by an MPO. 
This leads to entanglement distillation very similar to the i.i.d.\ case\all{;} \ja{surprisingly, }
%numerical evidence suggests that 
allowing for principally unwanted correlations between subsequent pairs can even speed up the convergence to maximally entangled pairs compared to the uncorrelated i.i.d.\ case. 
We \all{discuss} a simple physical example where this is the case for a large parameter region \ja{of initial fidelities and correlation strengths}.

\subsection*{Setting and formalism}  
We consider a sequence of $L$ pairs of qubits, where two parties (say Alice and Bob) each hold one qubit from each pair. 
The focus on qubits is set for simplicity of notation only, it is clear
that the same framework can be applied \all{to} systems of other physical dimension.
These pairs are entangled, as well as correlated with each other, as a consequence of the preparation procedure involving stationary quantum memory effects. A natural preparation 
exhibiting such a memory involves an auxiliary quantum system $C$ of some dimension $d$
that embodies all the degrees of freedom of the memory. The state is then prepared in a sequential fashion, with the
memory unitarily interacting with the first entangled pair, then the second, and so on \cite{Sequential,PlenioMemory,MPSSurvey}. 
A state generated in this way is \ja{given by} a matrix-product state, if it is pure, or a \emph{matrix-product
operator} in case of noisy mixed states \cite{Mixed,Zwolak}, as they are considered here, with $d$ taking the role of the \emph{bond dimension}. The decay of memory effects in the distance between the entangled pairs naturally emerges in this construction. \ja{We introduce here how }%It is the setting of 
naturally correlated bi-partite MPO \ja{arising from this setting.} %arising that is being introduced here.

More specifically, we work in a numerically indexed Bell basis $\left( \ket{\phi_1}, \ket{\phi_2}, \ket{\phi_3}, \ket{\phi_4}\right)$, more commonly labelled as $\left( \ket{\phi^+}, \ket{\phi^-}, \ket{\psi^+}, \ket{\psi^-} \right)$.   We consider a sequence of $L$ pairs of qubits, with basis vectors $\ket{ \Phi_{\mathbf{x}} } = \ket{ \phi_{x_1} }\ket{ \phi_{x_2} }\ldots\ket{ \phi_{x_L} }$, where Alice holds the first qubit of each pair and Bob holds its partner.  Translationally invariant mixed states reflecting stationarity of the source are described in the MPO language as
\begin{equation} 
 \bra{ \Phi_{\mathbf{x}}}  \rho   \ket{ \Phi_{\mathbf{y}}}=\Trace{  M^{x_1,y_1}M^{x_2,y_2} \ldots M^{x_L,y_L} }.
 \end{equation}
\all{Purely for simplicity of notation, we take periodic boundary conditions here.}
The dimension of the matrices $M^{x,y}\in \cc^{d\times d}$, $x,y\in\{1,\dots, 4\}$ 
limits the correlations between pairs, and by increasing this bond dimension $d$,
arbitrary quantum states can be described in this formalism.  There is a gauge freedom in our choice of MPO matrices as for any invertible $\all{S}$, 
mapping 
\begin{equation}
	M^{x,y}\mapsto \all{S} M^{x,y} \all{S}^{-1}
\end{equation} 
	will give an alternative description of the same physical state. 
Generally, 16 matrices are needed for the description of each pair, 
\all{reflecting the two-particle density matrix.}
However, without loss of generality we can assume
each $M^{x,y}$ to be Bell diagonal. If the state was not Bell diagonal originally, it can be brought into this form by 
a suitable local group twirl over the Pauli group \cite{Bennett}.
Since the employed protocols \all{only}
make use of Clifford operations, the group twirl will commute with these operations, so that it can be implemented at the very end or
merely at the level of classical data processing. For this reason we \all{use} the shorthand $A=M^{1,1}$, $B=M^{2,2}$, $C=M^{3,3}$ and $D=M^{4,4}$. Without loss of generality, we consider the distillation of maximally entangled $\phi^+$ pairs.  The ``$A$" matrix will be the dominant matrix, \all{while} the others we will call noise matrices.

\begin{figure}[b!]
\includegraphics[width=0.51\textwidth]{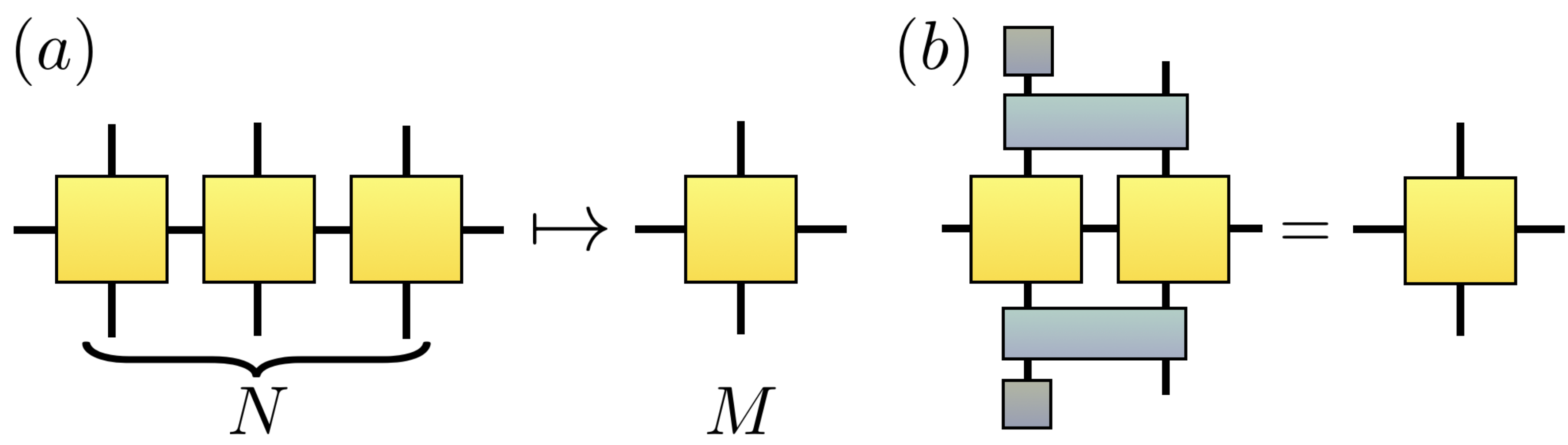}
\caption{Renormalisation schemes of MPO, (a) mapping $N$ pairs to $M$ in an $N\rightarrow M$ scheme. (b) In the $2\rightarrow 1$ recurrence protocol
two neighbouring MPOs are conjugated with a local unitary and subjected to a measurement. Contraction of the tensor network leads to the MPO at the 
subsequent scale.
}
\label{fig:ren}
\end{figure}

\subsection*{Protocols and renormalisation}

An $N \rightarrow M$ iterative protocol for entanglement distillation of i.i.d.\ states will act on $N$ pairs at a time and output $M$ (where $M<N$) pairs. For 
more than $N$ i.i.d.\ pairs, the protocol is performed in parallel on blocks of $N$ pairs. In the MPO setting pairs are not i.i.d.\ 
and so we must specify which pairs are involved in each block of a protocol.   We choose neighbouring pairs so the first $N$ pairs are distilled into $M$ pairs, 
while simultaneously the next $N$ pairs are distilled, and so on.  This natural choice has the practical merit of respecting locality, and has the additional
advantage that the output state is easily shown to again be an MPO of the same bond dimension (see Fig.~\ref{fig:ren}).   
Every iteration of the distillation protocol now acts as a map from an MPO on one scale to the subsequent 
one and reducing the chain length from $L$ to $L {M}/{N}$. After each step, a positive MPO is retained 
\cite{MPOPositive}. Indeed, 
it can be naturally seen as process of \emph{MPO renormalisation}, this being a mixed-state and bi-partite analogue of the renormalisation of matrix product states
discussed in Ref.\ \cite{RenormalisingMPS}. After the $n$-th step, we label the MPO operators $\{ A_n, B_n, C_n, D_n\}$, where the initial raw state provides the $n=0$ matrices.  We prove several results on convergence to entangled states which show the functioning of the schemes; proofs that can also be interpreted as convergence proofs for renormalisation flow of the MPOs. 
Specifically we consider the recurrence protocol and a distillation through error correction using the 5-qubit code.  Both protocols  rely solely on Clifford operations \all{and} use local operations and classical communication with respect to the partition between Alice and Bob. Intuitively, one can say that in many practically relevant settings, the entanglement and correlations between pairs are
being ``renormalised'' into more useful entanglement shared between Alice and Bob, to be employed in subsequent key generation.
 
\subsection*{Recurrence protocol} 
The recurrence protocol is a $2 \rightarrow 1$ iterative protocol 
which uses post-selection.  At every round measurement outcomes are being produced
and we only proceed if certain outcomes are obtained. Here we use a slightly improved version 
\cite{DeutschEkert} of the recurrence scheme \cite{BennettBrassard}. 
Cast into the MPO language, 
the iteration formula \all{after two steps} is 
%$(A_{n},B_{n},C_{n},D_{n})\mapsto(A'_{n+1},B'_{n+1},C'_{n+1},D'_{n+1})$
\begin{eqnarray}
	A_{n+2} &=  (A_{n}^{2}+ B_{n}^2 )^{2}+( C_{n}^2 + D_{n}^2)^2 , \\
	B_{n+2} &=   \{ A_{n} , B_{n} \} ^2 + \{ C_{n} , D_{n} \} ^2   , \\
	C_{n+2} &=   \{ A_{n}^2 + B_{n}^2  ,  C_{n}^2 + D_{n}^2 \} ,  \\
	D_{n+2} &=    \{  \{ A_{n} , B_{n} \} , \{ C_{n} , D_{n} \}  \} ,
%	\label{eqRec}
\end{eqnarray}
\ec{where curly brackets denote the anti-commutator.} \all{After two steps, the matrices are being re-gauged and rescaled.}  Replacing matrices by \ja{commuting }%commutative 
scalars recovers the original i.i.d.\ result.  Proving convergence of the initial MPO in the state $\rho$ to the maximally entangled state  $\phi^+$ is achieved by showing that the noise matrices vanish exponentially faster than the $A_{n}$ matrices.  Specifically, an appropriate ratio of norms will be exponentially suppressed.  We define a norm $\|M\|$ in terms of a channel $\mathcal{M}$ Choi-isomorphic to $M$, so that $\|M\|= \| \mathcal{M} \|_{1 \rightarrow 1}$, \all{where} we use the induced ``1-to-1'' %Schatten 
norm~\cite{bhatia}. 
We introduce the noise contribution of the coefficient matrices $B_{\all{n}}$, $C_{\all{n}}$, and $D_{\all{n}}$ 
as 
\begin{equation}
	\epsilon_n = \max \big(\toNorm{{\cal B}_n}{1}, \toNorm{{\cal C}_n}{1}, \toNorm{{\cal D}_n}{1}\big).
\end{equation}
Due to norm sub-multiplicativity, one finds \ec{initially small $\epsilon_0$ entails $\epsilon_n$ vanishes with $n$}.  However, ensuring $A_{n+2}$ stays large is difficult. \all{To do so, we shall adjust the MPO gauge after two steps,
re-gauging this using a suitable gauge transformation and re-scaling,}
so that ${\cal A}_{n+2}$ \all{is}
trace-preserving and hence $\toNorm{{\cal A}_{n+2}}{1}=1$. 
\all{To quantify how much the gauge transformation 
changes the matrix norm, we rely} on the 
\emph{ergodicity coefficient} of the matrices, 
\begin{equation}
 \tau(\mathcal{M}) = \max_{\Trace{\sigma} =0} \frac{\|\mathcal{M}\kl{\sigma}\|_1}{\|\sigma\|_1},
\end{equation}
which allows a quantification of how rapidly a channel mixes input states into the channel's stationary state. 
We are interested in the ergodicity of $\mathcal{A}_{n}$, for which we use the shorthand $\tau_{n}:= \tau(\mathcal{A}_n)$. We are now ready to state the first %of our 
main result, \ja{which provides} sufficient conditions for entanglement distillation using the recurrence protocol.

\begin{figure}[t!]
\includegraphics[width=0.48\textwidth]{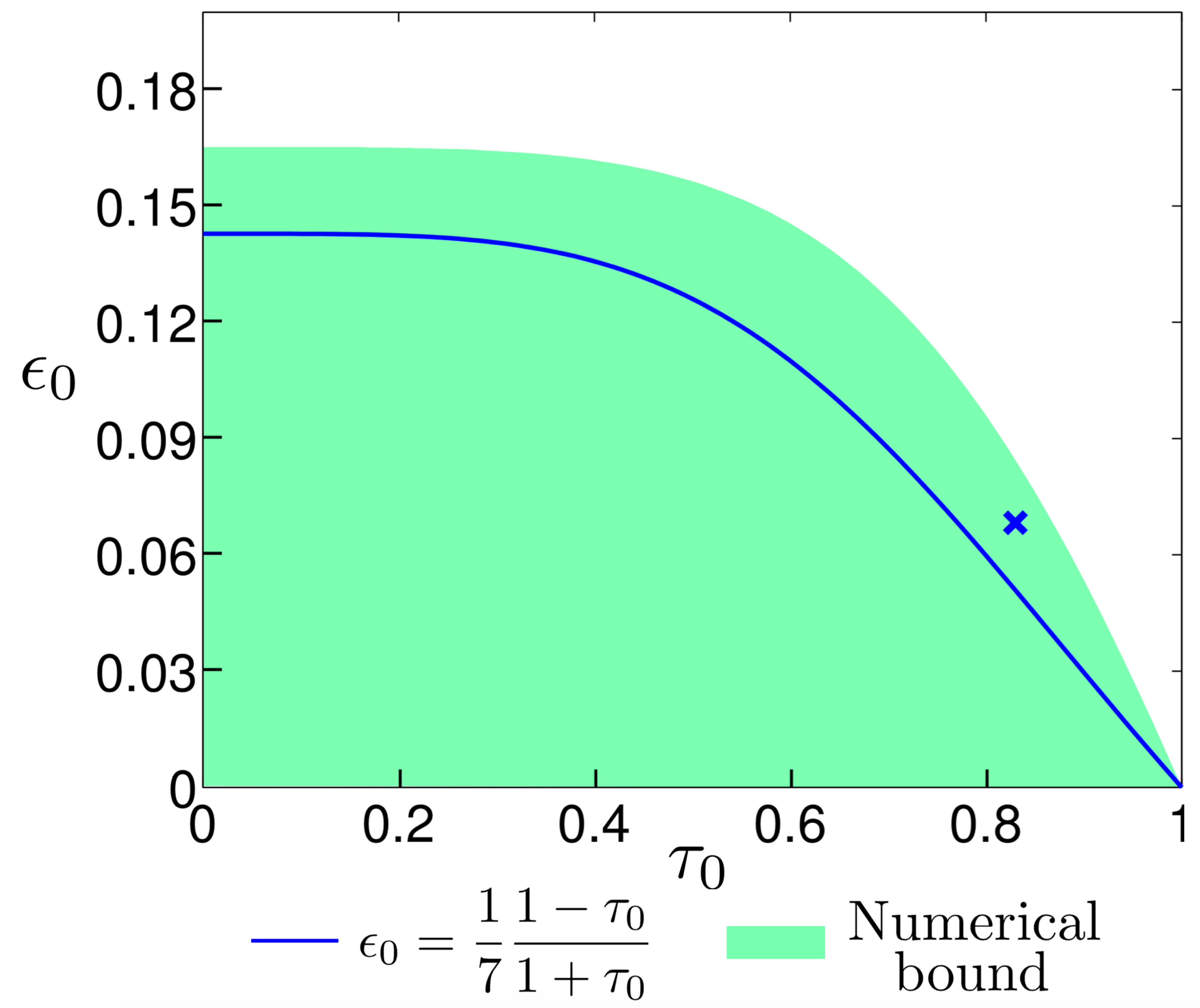}
\caption{Region of convergence for the recurrence protocol. The area under the blue line is the region fulfilling the conditions given in Theorem 1.  The green region is a slightly
improved bound that can be obtained with computer assistance, but for which we have no closed form expression. The blue dot corresponds to our physical example.}
\label{fig:taueps}
\end{figure}

\begin{thm}[Convergence in the recurrence protocol]\label{thm:recurrence}
	Given a translationally invariant \je{Bell} diagonal MPO with coefficient \ja{matrices} $A_0,B_0,C_0,$ and $D_0$, the iterative application of the recurrence protocol leads to convergence to uncorrelated pairs in the \all{maximally} entangled state $\phi_+$ for 
\begin{equation}	
		\epsilon_0 \leq \frac{1}{7} \frac{1-\tau_0^4}{1+\tau_0^4}.
\end{equation}
\end{thm}

The %conditions for 
convergence \ja{is} illustrated in Fig.~\ref{fig:taueps}.  The proof \je{is presented} in full length in the appendix. \all{To be self-contained,}
we will sketch it here, %however, 
and provide significant intuition. 

To \all{prove} 
convergence, we need to show that the noise matrices go down exponentially fast, while $A_{n}$ stays large. The first part can be shown by taking into account a double step of the protocol after which all norms of the noise matrices are at least of order $\eee_n^2$. This can be shown via sub-additivity and sub-multiplicativity of the 1-to-1 norm,
\begin{equation}\label{EqEpsilon}
 \je{\eee_{n+2}} \leq 4\kl{1+\eee_n^2}\eee_n^2 \je{.}
\end{equation}
However, to ensure the contribution of the dominant matrix stays large, our approach is to regauge so that $\jg{\mathcal{A}}_{n+2}$ is trace-preserving. A channel $\mathcal{A}_n$ is trace-preserving if and only if \ja{its Hilbert-Schmidt adjoint,} the dual channel $\mathcal{A}_{\all{n}}^{\dagger}$ (the Heisenberg representation of $\mathcal{A}_n$) satisfies $\mathcal{A}_n^{\dagger}(\id)=\id$.  When instead $\mathcal{A}_n^{\dagger}(\xi_n)=\lambda_n \xi_n$ (where $\lambda_n$ is the largest such eigenvalue), then the Perron---Frobenius theorem ensures $\xi_n$ is Hermitian.  \jg{The gauge transformation $\all{S}_n=\sqrt{\xi_\jg{n}}\otimes \sqrt{\xi_\jg{n}}$ and the \all{re-scaling by} $ \lambda_n^{-1}$  will recover trace-preservation, provided 
\je{$\xi_n$} is invertible.}  This transformation potentially increases the norm of the noise matrices.  Using sub-multiplicativity, we find that 
\begin{equation}
	\eee_{n+2} \mapsto \jg{\lambda_{n+2}^{-1}} \kappa\kl{\ja{\MS}_\jg{n+2}}\eee_{n+2}, 
\end{equation}
where $\kappa(\ja{\MS}_n)= \toNorm{\ja{\MS}_n}{1} \toNorm{\ja{\MS}_n^{-1}}{1}$ is the condition number of $\ja{\MS}_n$.  
%The factor $\lambda_n^{-1}$ can be dismissed as one easily 
%finds $\lambda_n^{-1} \leq 1$. 
We wish to show that $\mathcal{A}_{n+\jg{2}}$ stays close to trace-preserving to keep the condition number small. \jg{We show that}
%\begin{equation}
%    \kappa(\ja{\MS}_{n+2}) \leq \frac{1}{1-2k_{n+2}},\,\, 
%    k_{n+2}=\frac{1+\tau_{n+2}^4}{1-\tau_{n+2}^4}\left(4 \epsilon_{n+2}^2 +10 \epsilon_{n+2}^4 \right).
%\end{equation}
\jg{
\begin{equation}
	\kappa(\ja{\MS}_{n}) \leq \frac{1}{1-2k_{n}},\,\, 
    k_{n}=\frac{1+\tau_{n}^4}{1-\tau_{n}^4}\left(4 \epsilon_{n}^2 +10 \epsilon_{n}^4 \right).
	\end{equation}}
	The proof bears similarities to the perturbation of the steady state of a trace-preserving \emph{quantum Markov chain}
	\cite{Gudder},
which also depends on the ergodicity coefficient of the \jg{transition matrix}.  The basic intuition is that if $\mathcal{A}_n$ is a rapidly mixing channel, with small $\tau_{n}$, then $\mathcal{A}^4_n$ is also rapidly mixing. Before making our gauge transformation $A_{n+2}$ is a sum of ${A}^4_n$ and some small noise matrices.  The more rapidly mixing a channel, the more its eigenstates are robust against the perturbative addition of noise matrices.  We desire the dual eigenstate 
$\xi_{n+2}$ \all{before the gauge} to stay close to $\id$, which we expect provided for rapid mixing (small $\tau_{n}$) 
and low noise (small $\epsilon_{n}$) as we show rigorously in a spirit similar to the ideas of \je{Ref.} \cite{szehr_perturbation_2013}.

\ec{Although ergodicity is not a matrix norm, it has similar properties such as sub-additivity and sub-multiplicativity, from which one can derive an upper bound on $\tau_{n+2}$ in terms of $\tau_{n}$ and $\epsilon_{n}$ so that $\tau_{n+2} \leq \tau_n^4 + f(\epsilon_n, \tau_n)$ where $f$ is an slight correction that vanishes as $\epsilon_n \rightarrow 0$.  This occurs because a double step of the protocol raises the ergodicity coefficient of $A$ to the fourth power, but regauging the MPO results in the adjustment $f$.  The essential point is that we have bounds on the pair $(\epsilon_{n+2}, \tau_{n+2})$ in terms of  $(\epsilon_{n}, \tau_{n})$.  From these it is straightforward to numerically determine the region of convergence to the fixed point $(0, 0)$, which we show in Fig.~(\ref{fig:taueps}).  Also shown in this figure is an analytic curve, for which we show $(\epsilon_{n}, \tau_{n}) \rightarrow (0,0)$ without numerical aid.  Finally, convergence in MPO operators again entails convergence of the density matrix.}

% Although ergodicity is not a matrix norm, it has similar properties such a sub-additivity and sub-multiplicativity, from which one can derive an upper bound on $\tau_{n+2}$ in terms of $\tau_{n}$ and $\epsilon_{n}$.  A double step of the protocol raises the \je{ergodicity} coefficient of ${\cal A}_{\all{n}}$ to the fourth power, adds the perturbation and sets the new gauge. Since the perturbation is small and exponentially decaying, we can bound $\tau_{n+2}$ 
% \all{after the gauge}
% \begin{align}
%  \tec{\AAnn} \leq \tec{\AAn}^4 + \order{\eee_n},
%  \end{align}
% and thus show that $\tec{\AAn}$ decreases exponentially as well. 

\subsection*{5-qubit protocol} 
For every error correcting code that encodes $k$ qubits into $N$ physical qubits, there exists an iterative $N \rightarrow k$ one-way entanglement distillation protocol~\cite{Bennett}. In these protocols, noise information is extracted by local measurements, but instead of post-selecting when errors are detected we attempt to correct them by determining the smallest weight error consistent with the measurement data.  The advantage over the recurrence protocol is that this protocol is deterministic, and that one-way distillation schemes require much less classical communication.  In particular, we consider one-way entanglement distillation using the 5-qubit code (so a $5 \rightarrow 1$ protocol), which is the smallest known code capable of correcting any single qubit error.  We again find a closed form for the map acting on the coefficient matrices in each iteration, though we omit it here as each expression contains $4^4$ terms. We further introduce the \all{transfer} matrix $E_n=A_n+B_n+C_n+D_
n
$\all{with corresponding channel ${\cal E}_n$}, 
as it has useful properties over the course of the iterations in a 
deterministic protocol, and state the following theorem.

\begin{thm}[Error correction] Given a translationally invariant \je{Bell} diagonal MPO \all{with} 
coefficient \ja{matrices} $A_0,B_0,C_0,$ and $D_0$, the iterative application of the 5-qubit error correcting code leads to convergence 
\all{to} uncorrelated pairs in the \all{maximally entangled} state $\phi_+$ for $\eee_0 \leq {1}/{33}$. $\eee_0$ is defined as above in the gauge where $\mathcal{E}_0$ is trace-preserving.
\end{thm}

Again, the full proof \ja{is provided }% will be appear 
in the appendix. Similar to the post-selective case, we show that with a growing number of iterations $n$, the contribution of the dominant matrix $A_n$ \all{grows} exponentially faster than the contribution of the noise matrices. In the deterministic case, we can use our gauge freedom 
to make ${\cal E}_n$ \all{trace-preserving}. 
Since we do not post-select, we do not need to renormalise \all{in} every round.  In the 5-qubit error correction code, the 
\all{transfer} matrix is always mapped to its fifth power, $E_{n+1}=E_n^5$. Thus, if we initially choose a gauge where the corresponding map $\mathcal{E}_0$ is trace-preserving, the \all{transfer} matrix will keep this property over the course of the iteration. Using combinatorical arguments, we then prove that for suitably small $\eee_0$, $\eee_n$ converges to zero, entailing convergence in fidelity of the physical state. 

\subsection*{Numerical studies and physical Hamiltonians} 
To complement the rigorous and analytical results presented above, 
we have also performed numerical studies on randomly drawn matrices. These results \je{demonstrate}
that a state within the distillable region is guaranteed to converge, but \je{several} 
steps in our proof make worst case assumptions and so we may also find convergence for many MPO states outside this region. In \all{a} worst case  scenario, correlations are always \all{pernicious}.
However, our numerics indicate that in many strongly correlated chains the correlations can also be beneficial and enable distillation at noise levels well above the rigorous threshold as well as in cases where the protocols do not converge for uncorrelated i.i.d.\ states. 

%It should also be clear that a plethora of physical Hamiltonians reflecting actual interactions in meaningful preparation 
%procedures are covered by our results. 
\ja{It should be clear that many physical Hamiltonians reflecting interactions with memory channels in sequential preparation procedures are covered by our results.}
To present a \je{paradigmatic} 
example, we compare the performance of the recurrence protocol on (i) perfect memoryless i.i.d.\ distributions of $\phi^+$ states with (ii) sequentially prepared states with implemented memory. For the preparation of the memory, we prepare uncorrelated Werner states, %with a certain $F_0\in[0,1]$, 
which then undergo a unitary interaction $U(t,J) = \exp(itH)$ on Bob's side, where 
\je{the Hamiltonian} $H$ is given by
\begin{equation} 
H(J)=J\kl{X\otimes X+Y\otimes Y+Z\otimes Z} + Z\otimes \id + \id\otimes Z.
\end{equation}
We further implement a de-phasing channel to make the memory forgetful. We will discuss this example in further detail in the appendix. We identify parameter regions where the correlated states remarkably outperform the uncorrelated i.i.d.\ states regarding the speed of convergence. Also, distillation is possible for a larger range of initial fidelities. We are able to use the unwanted correlations from the memory to enhance the distillation, see also Fig.\ \ref{fig:taueps}. 

%In the post-selective recurrence protocol, we use a translationally invariant bell diagonal MPO with coefficient matrices $A, B, C, $ and $D$, with $A$ being trace-preserving. By bounding the ergodicity  $\tau(A)$ and the noise contribution $\epsilon \leq \{\|B\|,\|C\|,\|D\|\}$, the repeated application of the recurrence protocol on the MPO leads to convergence to the pure state $\phi^+$. Here, $\tau$ is the ergodicity coefficient introduced above. In fact, it is possible to find an upper bound on the noise for each ergodicity coefficient see figure. 

\subsection*{Perspectives} 
In this work, we have introduced a framework of renormalising entanglement in order to achieve entanglement distillation in the presence of natural correlations.
We have proven that protocols known to work for i.i.d.\ pairs above a threshold fidelity also give rise to feasible entanglement distillation.
We have identified criteria to ensure convergence of correlated pairs described by an MPO to a number of independent maximally entangled pure states. On intuitive grounds,
one might expect that if the MPO is only weakly correlated between the pairs and the reduced density matrix of a single pair is sufficiently close to a Bell pair, the distillation protocols should behave 
similarly to the i.i.d.\ case. Indeed, convergence can be proven for threshold fidelities and conditions on the correlation between the pairs. The programme initiated here shows that correlations
are not necessarily a disadvantage, and one does not have to aim at de-correlating pairs or resetting preparation procedures, steps that will take time and will in practice
lead to further entanglement deterioration. This work shows that such correlations can be largely renormalised away, with no modification to the schemes applied. We hope that this
work triggers further studies on entanglement distillation and repeater protocols in the presence of 
realistic memory effects, \je{as well as of further studies of renormalising matrix-product operators.}
  
\subsection*{Acknowledgements} 
We acknowledge funding from the BMBF (Q.com), the DFG (SPP ), the EU (SIQS, RAQUEL, AQuS), and the ERC (TAQ)
and comments from Martin Kliesch.

\section{5-qubit protocol}

\label{secErrCorr} 
\je{We now} prove a threshold for successful purification for the 5-qubit error correcting code used as an entanglement distillation scheme.
Instead of giving a lower bound on the fidelity, we give an upper bound on the infidelity, or the probability of measuring 
$\phi^-$, $\psi^+$ or $\phi^-$ for a pair of qubits.  

\begin{proof}
For $\phi^-$ we have  
\begin{equation}
 \je{p(\phi^-)}= \mathrm{Tr}[ \rho \left( \je{\phi^-} \otimes \id \right) ] = \frac{\Trace{B E^{L-1}}}{\Trace{E^L}} ,
\end{equation}
where the $\id$ acts on all other pairs of qubits.  Similar expressions hold for $\psi^+$ and $\psi^-$ by replacing $B$ with $C$ or $D$, respectively.
We will find an upper bound of this expression in terms of the channel norm of the associated observable transfer matrices $B$, $C$, and $D$. We assume locally purifiable
MPO's which allows us to make use of an isomorphism between the MPO matrices and completely positive maps. Specifically,
we define a norm $\|M\|$ in terms of a channel $\mathcal{M}$ Choi-isomorphic to $M$, so that $\|M\|= \| \mathcal{M} \|_{1 \rightarrow 1}$ and we use the induced
``1-to-1'' Schatten norm.  See Ref.~\cite{bhatia} and App.~\ref{Norms} for more on norms. We call $\EE$ the channel Choi-isomorphic to the transfer matrix $E$.
Our Lemma \ref{lem:NormBound} of the appendices proves that, assuming $\mathcal{E}$ is trace-preserving, we have
\begin{equation}
\label{FidBound}
\mathrm{Tr}[ \rho \left( \je{\phi^-} \otimes \id \right) ]  \leq \toNorm{\mathcal{B}}{1} \frac{1+d^{5/2}\tec{\EE}^{L-1}}{1-d^{5/2}\tec{\EE}^{L}}.
\end{equation}
Proving convergence of the initial MPO in the state $\rho$ to the maximally entangled state  $\phi^+$ is achieved by showing that the noise matrices vanish exponentially faster than the $A_{n}$ matrices. 

One step of the protocol takes 5 pairs as input and returns one pair as output, correcting all zero and one qubit errors and thus reducing the error probability to at least quadratic order in $\epsilon_n$.
As discussed before, the code applies without post-selection, which means that the state does not have to be renormalised. As the length of the chain $L_n$ is divided by five in every step, the new transfer matrix $E_{n+1}$ is simply the fifth power of the previous transfer matrix $E_{n}$. So if we start with a trace-preserving transfer channel, it remains trace-preserving. Thus, the renormalisation factor, derived in Lemma \ref{lem:TraceBound} is constant,
\begin{equation}
E_{n+1}=E_n^5 \Longrightarrow E_n=E_0^{5^n} 
\end{equation}
and
\begin{equation}
 L_{n+1} = \frac{L_n}{5} \Longrightarrow L_n=\frac{L_0}{5^n}.
\end{equation}
By sub-multiplicativity of the ergodicity, we have
\begin{equation}
\label{EnEvolve}
 \tec{\EEn}^L_n \leq \tec{\EE_0}^{5^n\frac{L_0}{5^n}} = \tec{\EE_0}^{L_0}.
\end{equation}
This already gives us an upper bound for the normalization for every step in the iteration just from the initial length of the chain and ergodicity.
The second step is to upper bound $\toNorm{\BBn}{1}$, $\toNorm{\CCn}{1}$ and $\toNorm{\DDn}{1}$.
We introduced, as a measure of the noise, the quantity
\begin{equation}
	 \eee_n=\max\skl{\toNorm{\BBn}{1},\toNorm{\CCn}{1},\toNorm{\DDn}{1}}.
  \end{equation}
  defined in a gauge and scaling, where $\EEn$ is trace-preserving. This gauge and scaling stays constant over the iteration. 
The complete update rules can be derived following the methodology of Sec.~\ref{secRecc}. We obtain a closed form for the map acting on the coefficient matrices in each iteration, though we omit it here as each expression contains $4^4$ terms and is not very insightful. Rather, we present just the leading terms here,
 \begin{eqnarray}
 A_{n+2}&=A_n^5 + A_n^4B_n + A_n^4C_n + A_n^4D_n + \dots,   \\
 B_{n+2}&=A_n^3B_n^2 + A_n^3C_nD_n + A_n^3D_nC_n + \dots,  \\
 C_{n+2}&=A_n^3B_nC_n + A_n^3C_nB_n + A_n^3D_n^2 + \dots,   \\
 D_{n+2}&=A_n^3B_nD_n + A_n^3C_n^2 + A_n^3D_nB_n + \dots  
\end{eqnarray}
Furthermore, we do not need the explicit iteration rule but only the norm of it. Each term is a product consisting of $A_n$ and some noise matrices. We can upper bound
the channel norm by using sub-multiplicativity and sub-additivity and the definition of our noise measure,
\jg{
\begin{align}
&\max	\left\{\toNorm{{\cal B}_{n+1}}{1},\toNorm{{\cal C}_{n+1}}{1},\toNorm{{\cal D}_{n+1}}{1}  \right\}\\
&~\leq 30\eee_n^2\toNorm{{\cal A}_n}{1}^3 + 70\eee_n^3\toNorm{{\cal A}_n}{1}^2+90\eee_n^4\toNorm{{\cal A}_n}{1}^2+66\eee_n^5.\nonumber
\end{align}
}
%\begin{align*}
%  \toNorm{B_{n+1}}{1}	&~\leq 30\eee_n^2\toNorm{A_n}{1}^3 + 70\eee_n^3\toNorm{A_n}{1}^2+90\eee_n^4\toNorm{A_n}{1}^2+66\eee_n^5, \\ 
%  \toNorm{C_{n+1}}{1}	&~\leq 30\eee_n^2\toNorm{A_n}{1}^3 + 70\eee_n^3\toNorm{A_n}{1}^2+90\eee_n^4\toNorm{A_n}{1}^2+66\eee_n^5, \\
%  \toNorm{D_{n+1}}{1}	&~\leq 30\eee_n^2\toNorm{A_n}{1}^3 + 70\eee_n^3\toNorm{A_n}{1}^2+90\eee_n^4\toNorm{A_n}{1}^2+66\eee_n^5.
%\end{align*}
\ec{The noise terms are at least of quadratic order.} The norm of $A_n$ can be easily upper bounded using the properties of channels and channel norms 
$\toNorm{\AAn}{1} \leq \toNorm{\EEn}{1} = 1$. The ensuing iteration for the noise measure is thus
\begin{equation}
    \eee_{n+1} \leq 30\eee_n^2+70\eee_n^3+90\eee_n^4+66\eee_n^5.
\end{equation}
Clearly, for errors to reduce $\eee_n$ must be at least smaller than $1/30$. In this regime, we have
\begin{align}
 \eee_{n+1} &\leq 30\eee_n^2+\frac{7}{3}\eee_n^2+\tel{10}\eee_n^2+\frac{11}{4500}\eee_n^2 
 = \kl{30+\frac{7}{3}+\tel{10}+\frac{11}{4500}}\eee_n^2 \leq 33\eee_n^2. 
\end{align}
So we can be sure the iteration converges if $\eee_0 \leq \tel{33} \approx 0.0303$. Numerically, 
we find that $\eee_0 \leq 0.031$.
The speed of the convergence is double \je{exponential} in the number of rounds,
\begin{equation}
 \eee_n \leq \tel{33}\kl{33 \eee_0}^{2^n}.
\end{equation}
Making use of Eq.~(\ref{FidBound}) and Eq.~(\ref{EnEvolve}), we find that infidelity is \je{bounded} by
\begin{equation} 
 \pophim{} \leq \tel{33}\kl{33 \eee_0}^{2^n} \frac{1+d^{5/2}\tec{\EE_0}^{L_0-1}}{1-d^{5/2}\tec{\EE_0}^{L_0}}.
\end{equation}
Since $|\tau(\mathcal{E}_0)| \leq 1$, we have in generic cases that in the limit of an infinite chain
\begin{equation} 
 \pophim{} \leq \tel{33}\kl{33 \eee_0}^{2^n},
\end{equation}
and it converges whenever $\toNorm{\je{\cal B}_0}{1}$, $\toNorm{\je{\cal C}_0}{1}$, $\toNorm{\je{\cal D}_0}{1} \leq \tel{33}
$, \je{which ends the proof.}
\end{proof}

It has hence been shown that we can formulate a threshold depending on the 1-to-1-norm of the three noise channels. This concludes our proof of a threshold for the five-qubit error correcting code.

\section{Recurrence protocol}

\je{We now turn to the post-selected recurrence protocol.}
It turns out that the proof is \je{significantly} more sophisticated since we can not rely on the trace-preserving property of the transfer channel. This requires a refined approach that involves deriving a perturbation bound for the left Perron vector of a quantum channel, \je{using ideas of Markov chain mixing.} 
 
\subsection{Computational vs. Bell basis}
\label{secRecc}

\begin{figure}[b!]
\includegraphics[width=0.35\textwidth]{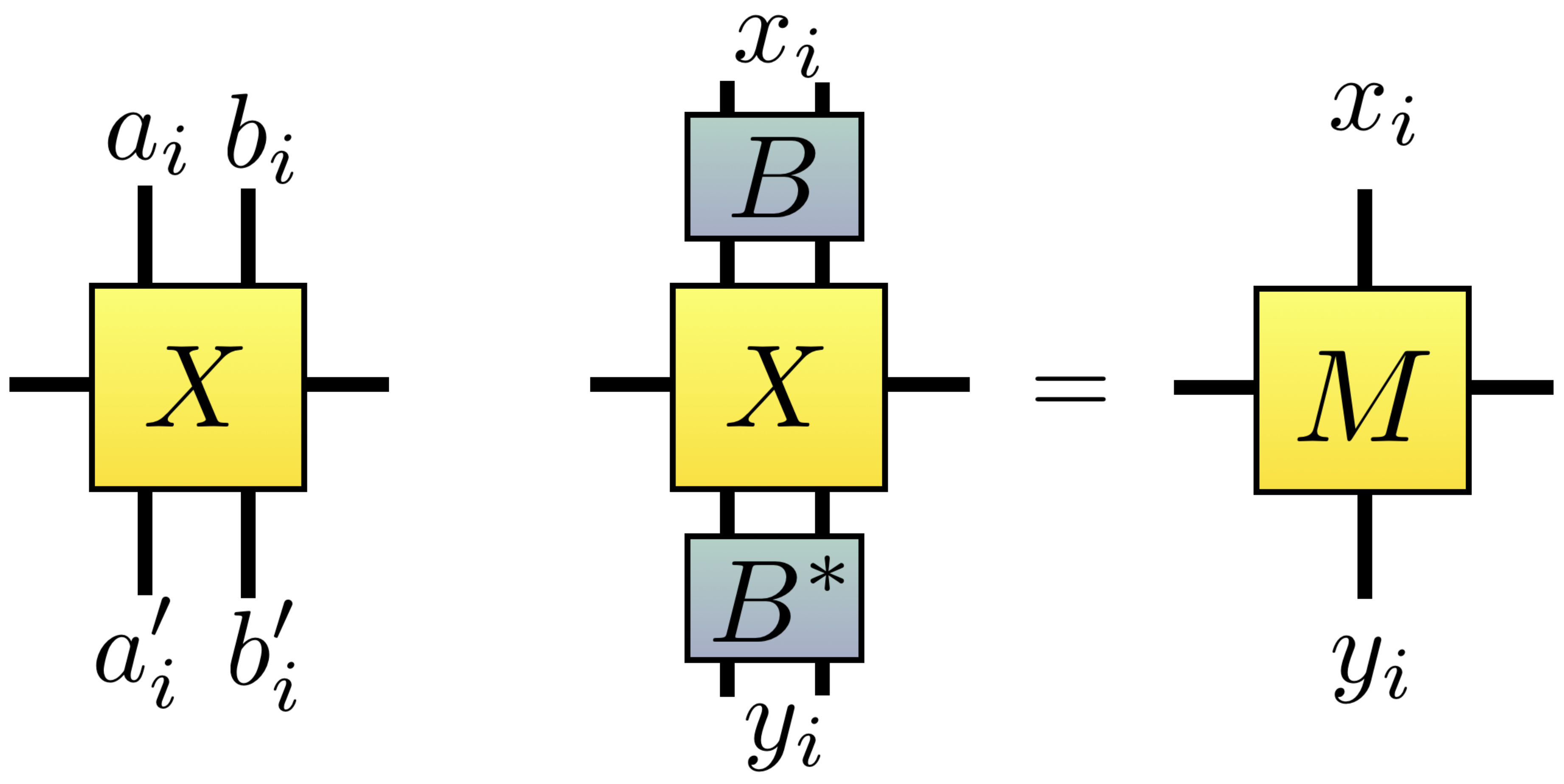}
\caption{A bi-partite matrix product operator (MPO), expressed in the Bell basis.}
\label{fig:mpo2}
\end{figure}

The recurrence protocol can be broken up into two steps performed locally by Alice and also Bob.  So Alice (and also Bob) divide their qubits up into adjacent pairs within the MPO chain.  For each pair the isometry $K=\vert 0 \rangle \langle 0,0 \vert + \vert 1 \rangle \langle 1,1 \vert$ is applied.  This can be implemented using a local CNOT, measurement, post-selection and disposing of the measured qubit.  The second step consists of a Hadamard rotation, again by both Alice and Bob.  It is easy to verify %(see appendix \ref{secMPO}) 
that after the first step we again get an MPO with the same bond dimension (see Fig.\,\ref{fig:mpo}) such that
\begin{equation}
	X^{x,y} \mapsto  (X^{x, y})^2.
\end{equation}
where the $X$ is the MPO operator in elements of the computational basis of 2 qubits, and so $x,y \in \{(0,0),(1,0),(0,1),(1,1) \}$. Indeed, this is direct analogous to the i.i.d case where the density matrix elements map as $\rho_{x,y} \mapsto \rho_{x,y}^2 $, though of course only in the computational basis. The phase noise is dealt with by the second step.  The bilateral Hadamard operation effectively swaps bit and phase flip noise so that both are 
dealt with. In the computational basis the bilateral Hadamard operation is unwieldily so we switch to the Bell basis (see Fig.~\ref{fig:mpo2}).  For the Bell basis we use the following shorthand
\begin{align}
	A = M^{1,1}= X^{(0,0),(0,0)}+X^{(1,1),(0,0)}+X^{(0,0),(1,1)}+X^{(0,0),(1,1)} , \\
    B = M^{2,2}= X^{(0,0),(0,0)}-X^{(1,1),(0,0)}-X^{(0,0),(1,1)}+X^{(0,0),(1,1)} , \\
    C = M^{3,3}= X^{(0,1),(0,1)}+X^{(0,1),(1,0)}+X^{(1,0),(0,1)}+X^{(1,0),(1,0)} , \\
    D = M^{4,4}= X^{(0,1),(0,1)}-X^{(0,1),(1,0)}-X^{(1,0),(0,1)}+X^{(1,0),(1,0)} ,
\end{align}
which only defines the MPO operators for the diagonal elements in the Bell-basis, but we will see that they are decoupled from other elements and so it is sufficient to consider these alone. In this basis the recurrence procedure implements
\begin{align}
	A_{n} &\mapsto A_{n}^2 + B_{n}^2 \mapsto   A_{n}^2 + B_{n}^2   = A_{n+1} , \\
	B_{n} &\mapsto  \{ A_{n} , B_{n} \}   \mapsto   C_{n}^2 + D_{n}^2  = B_{n+1} ,\\
	C_{n} &\mapsto  C_{n}^2 + D_{n}^2\mapsto  \{ A_{n} , B_{n} \}     = C_{n+1}, \\
	D_{n} &\mapsto  \{ C_{n} , D_{n} \}   \mapsto   \{ C_{n} , D_{n} \}   =  D_{n+1},
\end{align}
where each of the two steps are shown, and a subscript $n$ is introduced to denote the MPO after $n$ iterations.  The brackets $\{ \cdot , \cdot \}$ denote the anti-commutator.  For matrices that are simply scalars, where $B_{n}, C_{n} , D_{n} \sim \epsilon_{n}$ and $A_{n}\sim 1$, we see $B_{n+1}, D_{n+1} \sim 2 \epsilon^2$ but $C_{n+1} \sim O(\epsilon)$.  This occurs because in a single round only one type of noise is decreased, and so to see an overall $\epsilon^2$ error reduction we must consider two rounds of iteration
\begin{eqnarray}
	A_{n+2} &=  (A_{n}^{2}+ B_{n}^2 )^{2}+( C_{n}^2 + D_{n}^2)^2 , \\
	B_{n+2} &=   \{ A_{n} , B_{n} \} ^2 + \{ C_{n} , D_{n} \} ^2   , \\
	C_{n+2} &=   \{ A_{n}^2 + B_{n}^2  ,  C_{n}^2 + D_{n}^2 \} ,  \\
	D_{n+2} &=    \{  \{ A_{n} , B_{n} \} , \{ C_{n} , D_{n} \}  \} .
	\label{eqRec}
\end{eqnarray}
Now treating the matrices as scalars, we see  $B, C, D$ all go from size $O(\epsilon)$ to $O(\epsilon^2)$ or smaller.  This is the intuition from the i.i.d. case and we next turn to making this rigorous by quantifying this size with appropriate matrix norms. 
\subsection{Detailed outline of convergence proof for recurrence protocol}
\label{detailed}
In the main text we presented a sketch of the proof of \je{Theorem} \ ~\ref{thm:recurrence}.  We first recap this sketch, filling in some details, and clearly stating required lemmas.  Numerous technical tools relating to norms and the ergodicity coefficient are covered in Sec.~\ref{Norms}.  %Many readers may find it useful to first familiarise themselves with the simpler proof for deterministic protocols as detailed in Sec.~\ref{secPhysCon}.
The first step is showing the iterative formulae relating MPO operators after
$n+1$ distillation rounds as a polynomial of MPO operators after $n$ rounds, as
introduced in Eq.~\ref{eqRec}. \jg{We introduce a prime in the outcomes of Eq.~(\ref{eqRec}), $A_{n+2}',B_{n+2}',C_{n+2}',D_{n+2}'$, and we introduce the noise measure  
\begin{equation}
\epsilon_n' = \max \left(\toNorm{\mathcal{B}_n'}{1}, \toNorm{\mathcal{C}_n'}{1}, \toNorm{\mathcal{D}_n'}{1}\right). 
\end{equation}
In a second step, since the channel $\mathcal{A}'_{n+2}$ is not trace-preserving, we make use of the following lemma.}
 
\begin{lm}[\je{Conjugation of quantum channels}]
	Let $\FF$ be a completely-positive \je{map}, with largest eigenvalue $\lambda$ and $\FF^{\dagger}(\xi)=\lambda \xi$,
	\je{$\xi=\xi^\dagger$}. 
	Let $\mathcal{S}$ be a channel mapping $\mathcal{S}(\rho)=\xi^{1/2}\rho \xi^{1/2}$.  If $\xi$ is invertible, $\mathcal{S}^{-1}$ exists and $\mathcal{S} \circ \FF \circ \mathcal{S}^{-1}$ is trace-preserving.
\end{lm}
We give further steps of the proof in Sec.~\ref{secTP}.  Applying the lemma to $\mathcal{A}'_{n+2}$, gives a gauge transform we herein label $\mathcal{S}_n$. More specifically, we use
\jg{
\begin{align}
	A_{n+2} = \lambda_{n+2}^{-1} S_{n+2}A_{n+2}'S_{n+2}^{-1},\\ 
	B_{n+2} = \lambda_{n+2}^{-1} S_{n+2}B_{n+2}'S_{n+2}^{-1}, \\
	C_{n+2} = \lambda_{n+2}^{-1} S_{n+2}C_{n+2}'S_{n+2}^{-1}, \\
	D_{n+2} = \lambda_{n+2}^{-1} S_{n+2}D_{n+2}'S_{n+2}^{-1}. 
\end{align}
The matrix norms are not gauge invariant and m}uch of the proof centres around obtaining an upper bound on 
\begin{equation}
\label{eqNoiseAgain}
    \epsilon_n = \max \left(\toNorm{\mathcal{B}_n}{1}, \toNorm{\mathcal{C}_n}{1}, \toNorm{\mathcal{D}_n}{1}\right),
\end{equation}
in this gauge where $\mathcal{A}_{n}$ is trace-preserving. \jg{From sub-multiplicativity, we know that $\epsilon_{n+2} = \kappa(\mathcal{X}_{n+2}) \lambda_{n+2}^{-1} \epsilon_{n+2}'$ and as} argued \je{in the main text}, we obtain a recursive relation
\begin{equation}
\label{eqEPIandCond}
    \epsilon_{n+2} \leq 4\kappa( \mathcal{S}_{n}) \kl{1+\eee_n^2}\eee_n^2
\end{equation}
where 
\begin{equation}
	\kappa(\mathcal{S}_{n})=\toNorm{ \MS_n }{1} \,  \toNorm{ \MS_n^{-1} }{1}
\end{equation}
is the condition number of $\MS_n$ \jg{and where we omitted $\lambda$, as $\lambda \geq 1$}. A substantial amount of our technical work goes into proving the following lemma.

\begin{lm}[\je{Condition number}]
\label{LemKappa}
The condition number of $\kappa(\mathcal{X}_{n})$ where $\mathcal{X}_{n}$ is our gauge change, is upper bounded by   
\begin{equation}
\label{lemCond}
    \kappa(\MS_n) \leq (1-2k_{n})^{-1}  \with k_n=\frac{1+\tau_n^4}{1-\tau_n^4}\left(4 \epsilon_n^2 +10 \epsilon_n^4 \right) .
\end{equation}
\end{lm}

\je{This we prove in Sec.~\ref{secCondLem}}\jg{, introducing} \je{Theorem} \ \ref{TracePreservingDeviation}\jg{, an eigenvalue perturbation theorem, which we prove in Sec.~\ref{secEVThm}.} Notice that the \jg{bound on the} condition number depends on both $\tau_{n}$ and $\epsilon_n$, and in turn $\epsilon_{n+2}$ is now upper bounded by a function of only $\epsilon_n$ and $\tau_n$.
We already understand the iterative behaviour of $\epsilon_n$, but not of $\tau_{n}$.  In Sec.~\ref{secErgodUpdate} we show that
\begin{lm}[\je{Upper bound to ergodicity coefficient}]
\label{TauUpdate}
The ergodicity coefficient $\mathcal{A}_{n}$ obeys
\begin{align}
    \tau_{n+2} \leq \tau_n^4\kl{1+\Delta_n} + \frac{1+\Delta_n}{1-\Delta_n}\kl{\Delta_n + 2\epsilon_n^2 + 5\epsilon_n^4}
\end{align}
where 
\begin{equation}
    \Delta_n = k_n/(1 -k_n)
\end{equation}
with $k_n$ as in Lemma~\ref{lemCond}.
\end{lm}
Notice that $\tau_{n+2}$ depends only on $\tau_n$ and $\epsilon_n$.  Therefore, these two lemmas provide a pair of coupled equations that provide $(\epsilon_{n+2}, \tau_{n+2} )$ as a function of $(\epsilon_{n}, \tau_{n} )$.  Therefore, it is straightforward to numerically study the initial conditions $(\epsilon_{0}, \tau_{0} )$ flow towards the desired point $(0,0)$, and this is presented in Fig.~\ref{fig:taueps}.  These results show that a state within this distillable region is guaranteed to converge, but many steps in our argument make worst case assumptions and so we may also find convergence for many actual MPO states with $(\epsilon_{0}, \tau_{0} )$ outside this region.

The \je{precise shape of the} distillable region is difficult to characterise analytically. 
However, we can analytically prove convergence on a slightly smaller region.
\begin{lm}[\je{Proof of convergence}]
\label{AnalyticConverge}
Given the iterative formulae for upper bounds on $\epsilon_n$ and $\tau_n$, we know that $(\epsilon_n, \tau_n )\rightarrow (0,0)$ whenever 
\begin{equation}
\epsilon_0 \leq \min \left( \frac{1}{7} \frac{1-\tau_0^4}{1+\tau_0^4} \right).
\end{equation}
\end{lm}
We give a full proof in Sec.~\ref{secAnaCon}. We have shown that the MPO noise matrices vanish, while the $\mathcal{A}_n$ remains trace preserving and so of constant norm. \jg{As $\epsilon_n \rightarrow 0$ we have $\AAn \rightarrow \EEn$ and we can conclude convergence in fidelity as follows from Lemma \ \ref{lem:NormBound}}.
\section{Mathematical concepts and proofs}

\je{In this section, we present the detailed proofs of the statements on the recurrence protocol
of the main text. 
Since this requires some preparation, we will first define and introduce all tools made use of.}

\subsection{Definitions and properties of norms}
\label{Norms}
\subsubsection{Properties of norms}

As norms for the states we use the Schatten $p$-norm, 
\begin{equation}
	 \norm{\rho}{p}=\kl{\sum_i \sigma_i(\rho)^p }^{1/p} ,
\end{equation}	
where $\sigma_i(\rho)$ is the $i$th singular value of $\rho$ \je{\cite{bhatia}}. For the maps we will use norms \jg{induced by} Schatten norms 
\begin{equation}
	\norm{\FF}{a\ra b}=\maxs{}\frac{\norm{(\FF{\sigma)}}{a}}{\norm{\sigma}{b}}.
\end{equation}
We can now use, that
\begin{equation}
	\norm{\sigma_1}{p}=\maxs{2}\frac{\Trace{\sigma_2\sigma_1}}{\norm{\sigma_2}{q}} \with \tel{p}+\tel{q}=1 \jg{\leftrightarrow} q=\frac{p}{p-1},
\end{equation}
and deduce the following lemma.

\begin{lm}[\je{Norm of the dual channel}]
\label{AdjointNorm}
Let $\FF$ be a map. Then, for all $r,s>1$
 \begin{equation}
 \norm{\FF}{r\ra s} = \norm{\FF\adj}{\frac{s}{s-1}\ra\frac{r}{r-1}}.
 \end{equation}
\end{lm}

Here we use $\FF^{\dagger}$ to denote the \je{dual channel, .i.e., the} unique channel such that for all $A, B$ we have 
\begin{equation}
	\tr( A \FF(B))=\tr(  \FF^{\dagger}(A)B) .
\end{equation}

\begin{proof} This statement follows immediately from the definition of the norms. We have
 \begin{align}
\norm{\FF}{r\ra s} &= \maxs{1}\frac{\norm{\FF(\sigma_1)}{r}}{\norm{\sigma}{s}} 
= \maxs{1}\maxs{2} \frac{\Trace{\sigma_2 \FF (\sigma_1)}}{\norm{\sigma_2}{\frac{r}{r-1}}{\norm{\sigma_1}{s}}} \\
&= \maxs{1}\maxs{2} \frac{\Trace{\sigma_1 \FF\adj (\sigma_2)}}{\norm{\sigma_2}{\frac{r}{r-1}}{\norm{\sigma_1}{s}}}\nonumber \\
&= \maxs{2} \frac{ \norm{\FF\adj(\sigma_2)}{\frac{s}{s-1}}}{\norm{\sigma_2}{\frac{r}{r-1}}}\nonumber \\
&= \norm{\FF\adj}{\frac{s}{s-1}\ra\frac{r}{r-1}},\nonumber
\end{align}
which is the statement to be proven.
\end{proof}
Using Lemma \ref{AdjointNorm}, we state that $ \toNorm{\FF}{1}=\toNorm{\FF\adj}{\infty}$. 

\begin{lm}[\je{Restriction to positive operators}]
 Let $\FF$ be a positive map, \je{then}
 \begin{equation}
  \max_{\sigma}{\normfrac{\FF}{\sigma}}=\max_{\sigma\geq0}{\normfrac{\FF}{\sigma}}.
 \end{equation}
\end{lm}

\begin{proof}
Assume that the maximum for $\FF$ is reached by $\sigma_{\rm max}=\sigma^{+} - \sigma^{-}$ where $\sigma^+,\sigma^{-}\geq 0$ and $\tr ((\sigma^+)^{\dagger} \sigma^-)=0$.  Though this matrix is potentially non-positive ($\sigma \neq 0$), we show that there always exists a non-negative matrix $\sigma_{\rm max}$ that also achieves the maximum.  A completely positive map $\FF$ preserves positivity, and so
\begin{align}
  \oneNorm{ \FF \kl{\sigma_{\rm max}}} &= \oneNorm{ \FF  \kl{\sigma^+}- \FF  \kl{\sigma^-}} 
	\leq \oneNorm{\FF \kl{\sigma^+}} + \oneNorm{ \FF  \kl{\sigma^-}} \je{,}
\end{align}
where we again use the triangle inequality.  For a positive Hermitian matrix $||M||_1 = \tr(M)$, and since $\sigma^\pm$ are positive and $\FF$ preserves positivity (it is a \je{completely positive} map), we infer
\begin{align}
  \oneNorm{ \FF \kl{\sigma_{\rm max}}} &\leq \Trace{ \FF \kl{\sigma^+}} + \Trace{ \FF  \kl{\sigma^-}} 
	= \Trace{ \FF \kl{\sigma^+} + \FF  \kl{\sigma^-}} \nonumber\\
  &= \Trace{ \FF \kl{\sigma^+ + \sigma^-}} \nonumber\\
	&= \oneNorm{ \FF  \kl{\sigma^+ + \sigma^-}}.
\end{align}
Additionally we have
\begin{equation}
  \oneNorm{\sigma_{\rm max}} = \sum_i{\bkl{\lambda_i\kl{\sigma_{\rm max}}}} = \oneNorm{ \sigma^+ + \sigma^-}.
\end{equation}
So we can conclude that
\begin{equation}
 \normfrac{\FF}{\sigma_{\rm max}} \leq \normfrac{\FF}{\sigma^+ + \sigma^-}.
\end{equation}
Therefore, there exists a strictly positive matrix $\sigma'_{\mathrm{max}}=\sigma^+ + \sigma^-$ that also 
\je{achieves} the maximum value.
\end{proof} 

\begin{co}[\je{Relationship between norms}]
\label{coM}
 If $\FF$ is a positive map then
\begin{equation}
  \toNorm{\FF}{1}=\infNorm{\FF\adj\kl{\id}}.
  \end{equation}
\end{co}

\begin{proof}
Since in the variational definition of the $1\ra1$-norm stated above, the maximum is always reached on a positive state, we can use the trace properties and 
the variational relations of the Schatten norms,
 \begin{align}
  \max_{\sigma\geq0}{\normfrac{ \FF }{\sigma}} &= \max_{\sigma\geq0}{\frac{\Trace{\FF\kl{\sigma}}}{\oneNorm{\sigma}}} 
					   = \max_{\sigma\geq0}{\frac{\Trace{\id\FF\kl{\sigma}}}{\oneNorm{\sigma}}} \nonumber\\
					   &= \max_{\sigma\geq0}{\frac{\Trace{\FF\adj\kl{\id}\sigma}}{\oneNorm{\sigma}}} \nonumber\\
					  & = \infNorm{\FF\adj\kl{\id}}					 .  
 \end{align}
 \end{proof}
%\je{[This is rather well-known, but this does not matter, does it?]} 
From this corollary, two further corollaries follow.

\begin{co}[\je{Norms for quantum channels}]
\label{NormOne}
 If $\FF$ is a positive and trace-preserving map then $\toNorm{\FF}{1} = 1$.
\end{co}
\begin{co}[\je{Norms for conjugations}]\label{cor:SimpleChannelNorm}
	For a \je{map} of the form $\mathcal{S}\kl{\rho} = \xi^{1/2} \rho \xi^{1/2}$ \jg{with $\xi^\dag=\xi$},
we have that
\begin{equation}
	\toNorm{\mathcal{S}}{1}=\infNorm{{\xi}^{1/2} \id{\xi}^{1/2}} = \infNorm{\xi}.
\end{equation}
\end{co}

\subsubsection{Bounds of gauge maps}
\label{secGaugeBounds}
\jg{The condition number is small provided $\xi$ is close to $\id$. Specifically, following Cor.~(\ref{cor:SimpleChannelNorm}) we can upper bound
\begin{align}
	\toNorm{\mathcal{S}}{1}&= \infNorm{\xi}, \\
	\toNorm{\mathcal{S}^{-1}}{1}&= \infNorm{\xi^{-1}},
\end{align}
Observe that if $\xi$ is close to the identity then $\xi - \id$ is small, and it is helpful to reformulate the bounds on $\xi$ in terms of $\Delta:=\infNorm{\xi-1}$.  We find
\begin{equation}\label{eq:xi1}
 \infNorm{\xi}=\infNorm{\id-(\id-\xi)} \leq 1+\infNorm{\id-\xi} = 1+ \Delta.
\end{equation}
Similarly, we can deduce
\begin{align}
 \infNorm{\id} &= \infNorm{\xi^{-1} - \xi^{-1}(\id-\xi)}, \\
 1 &\geq \infNorm{\xi^{-1}} - \infNorm{\xi^{-1}}\infNorm{\id-\xi}, \\
 &\Ra \infNorm{\xi^{-1}} \leq \tel{1-\Delta},\label{eq:xi2}
\end{align}
as well as
\begin{align}
	\infNorm{\id - \xi^{-1}} &= \infNorm{\xi^{-1}\left( \xi - \id \right)} 
	\leq \infNorm{\xi^{-1}}\infNorm{\id - \xi} = \frac{\Delta}{1-\Delta}.
\end{align}
We can deduce similar bounds for $\MS$. If we apply $(\id -S)$ on some state $\rho$, we get
\begin{align}
	&(\id - \MS)(\rho) = \rho - \xi^{1/2}\rho\xi^{1/2} 
	=(\id - \xi^{1/2})\rho + \rho(\id - \xi^{1/2}) - (\id - \xi^{1/2})\rho(\id - \xi^{1/2}).
\end{align}
We can use the H{\"o}lder inequality and the fact that $\infNorm{\id - \xi^{1/2}}\leq \infNorm{\id - \xi} < 1$ to find
\begin{align}
	\toNorm{\id - \MS}{1} &\leq 3 \toNorm{\id - \xi}{1} = 3 \Delta,\label{eq:SS1}\\
	\toNorm{\id - \MS^{-1}}{1} &\leq \frac{3\Delta}{1-\Delta}.\label{eq:SS2}
\end{align}
}
\je{These bounds will be used later.}

\subsection{Ergodicity coefficient and fundamental channel}
\label{secErgod}

We make frequent use of another concept originating from 
\je{the theory of}
Markov chains, which is the \emph{ergodicity coefficient}. It is a measure for 
 how close a quantum channel is to a projection onto its steady state.
 The ergodicity is defined as follows: If $\rho$ is the steady state of a completely positive map and 
 \begin{equation}
 \oneNorm{\sigma}=\Trace{(\sigma\adj\sigma)^{1/2}} 
\end{equation}
is the trace norm of $\sigma$, then
 \begin{equation}
  \tec{\FF} = \max_{\Trace{\sigma\adj\rho}=0} \frac{\oneNorm{\FF\kl{\sigma}}}{\oneNorm{\sigma}} .
 \end{equation}
 The ergodicity coefficient is similar to the second eigenvalue of the map, and in fact it is straightforward to see that it always upper bounds the second eigenvalue.
 The ergodicity coefficient is sub-multiplicative. If two maps $\FF_1$ and $\FF_2$ are trace preserving then
 \begin{equation}
  \tec{\FF_1\FF_2} \leq \tec{\FF_1}\tec{\FF_2} .
 \end{equation}
So the product of the two maps is at least as close to being a projector as both maps individually. Multiplying 
\je{a large number of}
quantum channels, each with an ergodicity smaller than \je{unity}, 
\je{then} eventually leads to a projection.  

\subsection{Preserving the trace}
\label{secTP}
The observation that $A_iA_{i+1} = (A_i S)(S^{-1}A_{i+1})$ 
\je{for any $S\in Gl(d,\cc)$} 
means that an MPO is not uniquely defined and offers a gauge freedom. Since we deal with unnormalised MPO, we also have the freedom of rescaling. There is a canonical gauge and scale corresponding to the transfer
 channel $\EE$ being trace preserving.
% This is called the left canonical form\cite{schollwoeck_density-matrix_2011}\cite{perez-garcia_matrix_2006}. The right canonical form indicates $\EE\adj$ is trace preserving.
% To arrive at this form, first scale the MPO matrices $A_i$ in such a way that the Perron eigenvalue of $\EE$ equals one. Let $\xi$ be the left Perron state of $\EE$,
% \begin{equation}
%  \EE\adj\kl{\xi} = \xi.
% \end{equation}
We apply the gauge transformation $A_i\je{\mapsto} 
\je{A^\prime_i}=\jg{S}\je{A_i} \jg{S}^{-1}$ with $\jg{S=\xi^{1/2}\otimes\xi^{1/2}}$,\jg{ where $\xi$ is the left Perron eigenstate of $\EE$}. Consequently,
 \begin{align}
  E &\je{\mapsto} E^\prime=\kl{\xi^{1/2}\otimes\xi^{1/2}}E\kl{\xi^{-1/2}\otimes\xi^{-1/2}} 
 \end{align}
 and
  \begin{align}  
  \EE\kl{\rho} & \mapsto \xi^{1/2}\EE\kl{\xi^{-1/2}\rho\xi^{-1/2}}\xi^{1/2}.
 \end{align}
 Now we are in a gauge where $\id$ is the left Perron state,
 \begin{align}\label{Gauge}
  \EE\adj\kl{\id} & \mapsto \xi^{-1/2}\EE\kl{\xi^{1/2}\id\xi^{1/2}}\xi^{-1/2} = \xi^{-1/2}\EE\kl{\xi}\xi^{-1/2}
 = \xi^{-1/2}\xi\xi^{-1/2} = \id.
 \end{align}
 
\subsection{Condition number}
\label{secCondLem}

Here we prove the condition number bound presented as \je{Lemma} \ref{LemKappa}. \jg{We use Eqs.~\eqref{eq:xi1} and \eqref{eq:xi2} to derive a} new expression for the condition number
\begin{equation}
	\kappa(\MS_n)=||\xi_n ||_{\infty}||\xi_n^{-1} ||_{\infty} \leq \frac{1 + \Delta_n}{1 - \Delta_n}.\label{eq:kappa}
\end{equation}
To proceed we need to upper bound $\Delta_n$, which we achieve using the following powerful result.

\begin{thm}\label{TracePreservingDeviation} (Eigenvector perturbation theorem)
Let $\FF_1$ and $\FF_2$ be completely positive maps with spectral radius 1, and $\FF_1$ is trace-preserving. If $\xi$ is the left Perron state for $\FF_2$, so that  $\FF_2^{\dagger}(\xi)=\xi$, then
\begin{equation}
  \infNorm{\id-\x} \leq \frac{k}{1-k}
\end{equation}
where
\begin{equation}
     k = \left( \frac{1 + \tau(\FF_1)}{1 - \tau(\FF_1)} \right) \toNorm{\FF_1-\FF_2}{1}.
\end{equation}
\end{thm} 
We prove this result in the following subsections, but here make direct use of it.  We set $\FF_1=\mathcal{A}_{n}^4$, which inherits the required properties from $\mathcal{A}_{n}$.  Likewise, we set 
\begin{equation} 
	\FF_2:=\lambda^{-1} (\AAn^4+ \PP_n), 
\end{equation}
where $\PP_n := \mathcal{A}_{n+2}' -\mathcal{A}_n^4 = \left\{ \AAn^2, \BBn^2 \right\} + \BBn^4 + \left( \CCn^2 + \DDn^2 \right)^2$ is a perturbation composed of noise matrices and $\lambda$ is a normalisation constant ensuring that $\FF_2$ has spectral radius 1.  Note that $\FF_2$ differs from $\mathcal{A}'_{n+2}$ by only a constant so they have the same eigenvectors, e.g. $\xi$.
Therefore, the theorem tells us that $\Delta_n \leq k_n /( 1 - k_n) $ where
\begin{equation}
\label{eqKnAppen}
    k_n = \left( \frac{1 + \tau( \mathcal{A}_n^4 )}{1 - \tau( \mathcal{A}_n^4 )} \right) \toNorm{ \mathcal{A}_n^4 - \lambda^{-1}(\mathcal{A}_n^4+ \PP_n)}{1} \je{.}
\end{equation}
Since $\mathcal{A}_n$ is trace-preserving $ \tau( \mathcal{A}_n^4 ) = \tau( \mathcal{A}_n )^4 = \tau_n$.  Looking at the second factor, we collect the $\mathcal{A}_n^4$ terms and use the triangle inequality
\begin{eqnarray}
\toNorm{\FF_1-\FF_2}{1} &=& \toNorm{(1-\lambda^{-1}) \mathcal{A}_n^4-\lambda^{-1} \PP_n}{1} 
 \leq  |1-\lambda^{-1}| \toNorm{\mathcal{A}_n^4}{1} + \lambda^{-1}  \toNorm{\PP_n^4}{1}.
\end{eqnarray}
To proceed, we need information about $\lambda$, which is the spectral radius of $\mathcal{A}_n^4 + \PP_n$ and so $\lambda \leq 1 + \toNorm{\PP_n}{1} $.  Furthermore, because $\mathcal{A}_n^4$ and $\PP_n$ are both positive channels, we know that $\lambda$ must exceed the spectral radius of $\mathcal{A}_n^4$ and so $1 \leq \lambda$.  Therefore,  $\lambda^{-1}\leq 1$ and $|1-\lambda^{-1}| \leq  \toNorm{\PP_n}{1} $. Combining these observations, we have
\begin{eqnarray}
\toNorm{\FF_1-\FF_2}{1} & \leq & 2 \toNorm{\PP_n}{1}\je{,}
\end{eqnarray}
and so
\begin{equation}
    k_n =2 \left( \frac{1 + \tau_n^4}{1 - \tau_n^4} \right) \toNorm{\PP_n}{1}.
\end{equation}
To upper bound $\toNorm{\PP_n^4}{1}$, we have to refer back to the iterative formulae and use norm sub-multiplicativity 
to show $\toNorm{\PP_n^4}{1} \leq 2 \epsilon_n^2 + 5 \epsilon_n^4$.  Substituting this into $k_n$ we get
\begin{equation}
    k_n = \left( \frac{1 + \tau_n^4}{1 - \tau_n^4} \right) (4 \epsilon_n^2 + 10 \epsilon_n^4),
\end{equation}
which proves \je{Lemma} \ref{LemKappa}.  However, the proof rests upon \je{Theorem} \ \ref{TracePreservingDeviation}, which we turn to in the next section.

\subsection{Proof of the eigenvector perturbation theorem}

Our methodology for proving \je{Theorem} \ \ref{TracePreservingDeviation} is in the spirit 
of \je{Ref.} \cite{szehr_perturbation_2013}, but \je{significantly}
generalised so that one of the channels need not be trace-preserving. The proof requires some new concepts we have not yet introduced, including the fundamental channel.
 \begin{df}[\je{Fundamental channel}]
Let $\FF$ be a channel, then the fundamental channel of $\FF$ is  
\begin{equation}
    \mathcal{Z} = \kl{\idd - \FF + \FFinf{}}^{-1}\je{.}
\end{equation}
\end{df}
This definition is central to the following two lemmas.

\label{secEVThm}
\begin{lm}[\je{Bound for fundamental channels}]
\label{lemTracePreservingDeviation} 
 $\FF_1$ and $\FF_2$ are completely positive maps with spectral radius 1, and $\FF_1$ is trace-preserving.
 $\ZZ_1$ is the fundamental channel of $\FF_1$. The left Perron state for $\FF_2$ is $\xi$.
 Then
 \begin{equation}
  \infNorm{\id-\x} \leq \frac{\toNorm{\ZZ_1}{1}\toNorm{\FF_1-\FF_2}{1}}{1-\toNorm{\ZZ_1}{1}\toNorm{\FF_1-\FF_2}{1}}
  \je{.}
 \end{equation}
\end{lm}
\noindent
\je{The following also holds.}

\begin{lm}[\je{Second bound for fundamental channel}]
\label{ZToTau}
 Let $\FF$ be a CPT-map, and denote $\ZZ$ to be the fundamental channel of $\FF$. It follows that
 \begin{equation}
 \toNorm{\ZZ}{1} \leq \frac{1+\tec{\FF}}{1-\tec{\FF}}.
\end{equation}
\end{lm}
Combining these results straightforwardly leads to \je{Theorem} \ \ref{TracePreservingDeviation}, and so the remainder of this section will prove these lemmas. 

\begin{proof}
We \je{begin} by bounding how much the left Perron state is perturbed from the identity. First we look at the projectors of our maps in the matrix picture,
\begin{align}
 \Finf{1}&=\ket{\rho_1}\bra{\id} \with \braket{\id}{\rho_1}=\Trace{\rho_1}=1, \\
 \Finf{2}&=\ket{\rho_2}\bra{\x} \with \braket{\x}{\rho_2}=\Trace{\x\adj\rho_2}=1.
 \end{align}
 Since the eigenvector matters only up to a constant we can rescale $\ket{\x}$ and $\ket{\rho_2}$ so they still satisfy $\SmallTrace{\x\adj\rho_2}=1$, but also satisfy $\braket{\rho_1}{\x}=\SmallTrace{\rho_1\adj \x}=1$.
 
Since we are dealing with the left eigenvectors we have to transpose our maps for easier notation.
Thus, the transposed projectors are
\begin{align}
 (\Finf{1})^{\dagger}&=\ket{\id}\bra{\rho_1} \with \braket{\id}{\rho_1}=1, \\
 (\Finf{2})^{\dagger}&=\ket{\x}\bra{\rho_2} \with \braket{\xi}{\rho_1}=1.
 \end{align}
 We start by applying $(\ZZ_1\adj)^{-1}$, 
 the inverse fundamental matrix of transposed $\FF_1$ to the difference between $\xi$ and the identity matrix $\id$.
 \je{This gives}
\begin{align}
 \inv{\ZZ_1\adj}\kl{\id-\xi}&=\kl{\idd - \FF_1\adj + \FFinf{1}\adj}\kl{\id-\xi} \nonumber\\
 &=\id - \id + \id - \xi + \FF_1\adj\kl{\xi} - \id\Trace{\rho_1\adj \xi} \nonumber \\
 &= \xi + \FF_1\adj\kl{\xi} \je{.}
\end{align}
Going from first to second line, we have used $(\FF_1)\adj( \id ) =\id$ and $(\FFinf{1})\adj( \id ) =\id$.  Going from second to third we use the normalisation condition $\SmallTrace{\rho_1\adj \xi}=1$ and cancel the identities.  It is a condition of the lemma that  $\xi = \FF_2^{\dagger}(\xi) $ and so
\begin{equation}
 \inv{\ZZ_1\adj}\kl{\id-\xi} =\kl{\FF_1\adj-\FF_2\adj}\kl{\xi}\je{.}
\end{equation}
Now we multiply by $\ZZ_1\adj$ on both sides and take the $\infty$-norm,
\begin{eqnarray}
 \infNorm{\kl{\id-\x}}&=&\infNorm{\ZZ_1\adj\circ(\FF_1\adj-\FF_2\adj)\kl{\x}} \\
				  &\leq & \frac{\infNorm{\ZZ_1\adj\circ(\FF_1\adj-\FF_2\adj)\kl{\x}}}{\infNorm{\x}}\infNorm{\x}
				  \nonumber \\ 
				  &\leq & \toNorm{\ZZ_1\adj\circ(\FF_1\adj-\FF_2\adj)}{\infty}\infNorm{\x} \nonumber
				 \\
				  &\leq & \toNorm{\ZZ_1\adj\circ(\FF_1\adj-\FF_2\adj)}{\infty}\kl{\infNorm{\id} + \infNorm{\id-\x}} 
				  \je{.}
				 \nonumber
\end{eqnarray}
We now make use of Lemma \ref{AdjointNorm} to deal with the adjoints and use that the identity has 
\je{unit} $\infty$-norm, \je{to get}
 \begin{equation}
 \infNorm{\kl{\id-\x}} \leq \toNorm{(\FF_1-\FF_2)\circ \ZZ_1}{1}\kl{1 + \infNorm{\id-\x}}
 \end{equation}
 Now we subtract the term $\toNorm{\kl{\FF_1-\FF_2}\circ \ZZ_1}{1}\infNorm{\id-\x}$ from both sides,
 \je{to arrive at}
 \begin{align}
  \infNorm{\kl{\id-\x}}&\kl{1-\toNorm{\kl{\FF_1-\FF_2}\circ \ZZ_1}{1}} 
	 \leq \toNorm{\kl{\FF_1-\FF_2}\circ \ZZ_1}{1} .
 \end{align}
Assuming that $1-\toNorm{\kl{\FF_1-\FF_2}\circ \ZZ_1}{1} \geq 0$, which is true for $\FF_1$ and $\FF_2$ being close enough we divide by $1-\toNorm{\kl{\FF_1-\FF_2}\circ \ZZ_1}{1}$,
\je{to get}
 \begin{align}
  \infNorm{\id-\x} \leq \frac{\toNorm{\ZZ_1}{1}\toNorm{\FF_1-\FF_2}{1}}{1-\toNorm{\ZZ_1}{1}\toNorm{\FF_1-\FF_2}{1}}
  \je{.}
 \end{align}
This completes the proof of \je{Lemma} \ref{lemTracePreservingDeviation}.
\end{proof}

Next we turn our attention to proving \je{Lemma} \ref{ZToTau}.  

\begin{proof}
First we reformulate the definition of the fundamental channel,
 \begin{align}
 \ZZZ{\FF}&=\kl{\idd - \FF + \FFinf{}}^{-1} =\sum_{k=0}^{\infty}{\kl{\FF-\FFinf{}}^k} \\
 &= \idd + \sum_{k=1}^{\infty}\kl{\FF^k-\FFinf{}} \nonumber \\
 &= \idd + \sum_{k=0}^{\infty}\kl{\FF-\FFinf{}}^k\kl{\FF-\FFinf{}}  \nonumber\\
 &= \idd + \sum_{k=0}^{\infty}\FF^k\kl{\FF-\FFinf{}}\je{.} \nonumber
\end{align}
In those steps we used that $\kl{\FF_1-\FFinf{}}^k=\FF_1^k-\FFinf{}$ and that $\FFinf{}\kl{\FF-\FFinf{}}=0$.
It is important to note that for any $\sigma$, the expression $\kl{\FF-\FFinf{}}\kl{\sigma}$ is traceless.
\je{This gives}
\begin{align}
 &\toNorm{\ZZ}{1} = \maxs{1}\sigmafrac{\kl{ \idd +\sum_{k=0}^{\infty}\FF^k\circ\kl{\FF-\FFinf{}}}}{1}  \\
	       &\leq \maxs{1}\sigmafrac{\idd}{1} + \sigmafrac{\sum_{k=0}^{\infty}\FF^k\circ\kl{\FF-\FFinf{}}}{1} \nonumber\\
	       &\leq 1 + \maxs{1} \normfrac{\sum_{k=0}^{\infty}\FF^k}{\FF-\FFinf{}\kl{\sigma_1}} \normfrac{\FF-\FFinf{}}{\sigma_1} \nonumber\\
	       &\leq 1 + \maxt{2}\normfrac{\sum_{k=0}^{\infty}\FF^k}{\sigma_2} \maxs{1}\normfrac{\FF-\FFinf{}}{\sigma_1}.
	        \nonumber
\end{align}
We now upper bound both terms separately. We note that $\FF-\FFinf{}=\FF\kl{\idd - \FFinf{}}$.
\je{We therefore get}
\begin{align}
&\maxs{1}\normfrac{( \FF-\FFinf{} )}{\sigma_1} \\
&= \maxs{1}\normfrac{\FF}{\idd-\FFinf{}\kl{\sigma_1}}\normfrac{\idd-\FFinf{}}{\sigma_1}  \nonumber\\
					  &\leq \maxt{1}\normfrac{\FF}{\idd-\FFinf{}\kl{\sigma_1}} \maxs{2}\normfrac{(\idd-\FFinf{})}{\sigma_1}  \nonumber\\
					  &\leq \tec{\FF} \toNorm{\idd-\FFinf{}}{1} \leq \tec{\FF} \kl{\toNorm{\idd}{1}+\toNorm{\FFinf{}}{1}}  \nonumber\\
					  &\leq 2\tec{\FF}. \nonumber
\end{align}
Here, we have used the fact that both $\idd$ and $\FFinf{}$ are trace-preserving and completely positive and Cor.~(\ref{NormOne}). \je{Hence,}
\begin{align}
 \maxt{2}&\normfrac{\sum_{k=0}^{\infty}\FF^k}{\sigma_2} \leq \sum_{k=0}^{\infty}\maxt{1}\frac{\oneNorm{\FF^k\kl{\sigma_1}}}{\oneNorm{\sigma_1}} \nonumber\\
 &=  \sum_{k=0}^{\infty}\tec{\FF^k} \leq \sum_{k=0}^{\infty}\tec{\FF}^k = \tel{1-\tec{\FF}}.
\end{align}
Thus we can conclude that
\begin{equation}
 \toNorm{\ZZ}{1} \leq 1+ \frac{2\tec{\FF}}{1-\tec{\FF}} = \frac{1+\tec{\FF}}{1-\tec{\FF}}.
 \end{equation}
This completes the proof of \je{Lemma} \ref{ZToTau}.
\end{proof}

%%%%%%%%%%%%%%%%%%%%%%%%%%%%%%%%%%%%%%%%%%%%%%%%%%%%%%%%%%%%%%%%%%%%%%%%%%%%%%%%%%%%%

\subsection{Change of the ergodicity}
\label{secErgodUpdate}
After re-gauging, the ergodicity coefficient changes. This change can be bounded in the following way.

\begin{lm}[\je{Bounding the change of the ergodicity coefficient}]\label{BoundingTau}
 Let $\FF_1$ and $\FF_2$ be trace-preserving, c.p.\ maps. Then
 \begin{equation}
 \bkl{\tec{\FF_1}-\tec{\FF_2}} \leq \tec{\FF_1-\FF_2}.
 \end{equation}
\end{lm}
\begin{proof}
Without loss of generality we assume $\tec{\FF_1}$ to be \je{larger} than $\tec{\FF_2}$.
\je{We then have}
\begin{align}
 |\tec{\FF_1}-\tec{\FF_2}|
 &=\maxt{1}\oneoneNorm{\FF_1}{\sigma_1}-\maxt{2}\oneoneNorm{\FF_2}{\sigma_2}  \\
			  &\leq\maxt{1}\oneoneNorm{\FF_1}{\sigma_1}-\oneoneNorm{\FF_2}{\sigma_1}\nonumber \\
			  &\leq\maxt{1}\oneoneNorm{(\FF_1-\FF_2)}{\sigma_1}
			  \nonumber\\
			  & =\tec{\FF_1-\FF_2}	.\nonumber		  
\end{align}
\end{proof}

Next, we still need to bound the ergodicity coefficient for $\MS_n\kl{\AAn^4+\PP_n}\MS^{-1}_n$. For this purpose we use Lemma \ref{BoundingTau} (in the second step),
\je{to get}
\begin{align}
 &\tec{\MS_n\kl{\AAn^4+\PP_n}\MS^{-1}_n} \nonumber \\
 &\leq \tec{\AAn^4} + \bkl{\tec{\AAn^4}-\tec{\MS_n\kl{\AAn^4+\PP_n}\MS^{-1}_n}} \nonumber \\
 &\leq \tec{\AAn}^4 + \tec{\AAn^4-\MS_n\kl{\AAn^4+\PP_n}\MS^{-1}_n}.
\end{align}
Subsequently, we bound the second term\je{,}
 \begin{align}
 \AAn^4&-\MS_n\kl{\AAn^4+\PP_n}\MS^{-1}_n 
 =\AAn^4-\MS_n \AAn^4\MS^{-1}_n - \MS_n \PP_n\MS^{-1}_n  \\
 &=(\idd-\MS_n)\AAn^4 + \MS_n \AAn^4 (\idd-\MS^{-1}_n) - \MS_n \PP_n\MS^{-1}_n.\nonumber
\end{align}
Next, we use the fact that $\tec{\FF_1+\FF_2} \leq \tec{\FF_1}+\tec{\FF_2}$ and that $\tec{\FF} \leq \toNorm{\FF}{1}$ in the following steps,
\begin{align}
 \tec{\AAn^4-\MS_n\kl{\AAn^4+\PP_n}\MS^{-1}_n} 
 &\leq \tec{(\idd-\MS_n)\AAn^4} + \tec{\MS_n \AAn^4 (\idd-\MS^{-1}_n)}  \\
 &  + \tec{\MS_n \PP_n\MS^{-1}_n} \nonumber\\
 & \leq \tec{\idd-\MS_n}\tec{\AAn^4} + \toNorm{\MS_n \AAn^4 (\idd-\MS^{-1}_n)}{1} \nonumber\\
 & + \toNorm{\MS_n \PP_n\MS^{-1}_n}{1}\nonumber \\
 & \leq \oneToOne{\idd-\MS_n}\tec{\AAn}^4 \nonumber\\
 &  + \oneToOne{\MS_n}\oneToOne{\AAn^4}\oneToOne{\idd-\MS^{-1}_n} \nonumber\\
 &  + \oneToOne{\MS_n}\oneToOne{\PP_n}\oneToOne{\MS^{-1}_n}.\nonumber
\end{align}
Using Eqs.\ \eqref{eq:xi1} to \eqref{eq:SS2} as well as \eqref{eq:kappa}, we conclude that
\begin{align}
 &\tec{\AAn^4-\MS_n\kl{\AAn^4+\PP_n}\MS^{-1}_n} 
 \leq \tec{\AAn}^4 3\Delta_n + \frac{1+\Delta_n}{1-\Delta_n}3\Delta_n + \frac{1+\Delta_n}{1-\Delta_n}\oneToOne{\PP_n}.
\end{align}
Therefore,
\begin{align}
 &\tec{\MS_n\kl{\AAn^4+\PP_n}\MS^{-1}_n} 
 \leq \tec{\AAn^4}(1 + 3\Delta_n) + \frac{1+\Delta_n}{1-\Delta_n} \left(3\Delta_n + \oneToOne{\PP_n}\right) 
 \je{.}
\end{align}
Next, we can upper bound the new ergodicity $\tec{\AAnn}$,
\je{to get}
\begin{align}
 \tec{\AAnn} &= \tec{\frac{\MS_n\kl{\AAn^4+\PP_n}\MS^{-1}_n}{1+\delta_R}} 
 \leq \tec{\MS_n\kl{\AAn^4+\PP_n}\MS^{-1}_n}.
\end{align}
It follows, from assuming the worst case of a growing ergodicity, that
\begin{align}
 &\tec{\AAnn} 
 \leq \tec{\AAn}^4\kl{1+3\Delta_n} + \frac{1+\Delta_n}{1-\Delta_n}\kl{3\Delta_n + \oneToOne{\PP_n}}.
\end{align}

\subsection{Simplifying parameter space}
\label{secAnaCon}
We present our bounds making use of \je{the} abbreviations $\tau_n\je{:=}\tec{A_n}$, $Z_n\je{:=}({1+\tau_n^4})/({1-\tau_n^4} )$ and $P_n\je{:=}\toNorm{\PP_n}{1}$. \je{These quantities can be bounded as follows,}
 \begin{align}
 P_n &\leq 2\eee_n^2+5\eee_n^4, \\ 
 \Delta_n &\leq \frac{2Z_nP_n}{1-2Z_nP_n}, \\
 \tau_{n+2} &\leq \tau_n^4\kl{1+3\Delta_n}+\frac{1+\Delta}{1-\Delta_n}\kl{3\Delta_n+P_n}, \label{eq:TauUpdate}\\ 
 \eee_{n+2} &\leq 4\frac{1+\Delta_n}{1-\Delta_n}\kl{\eee_n^2+\eee_n^4} .
\end{align}
To derive a manageable convergence threshold for $\eee_0$ depending $\tau_0$, we choose the ansatz
\begin{equation}\label{eq:ansatz}
 Z_n\eee_n \leq \tel{7}. 
\end{equation} 
Next, we show that given \eqref{eq:ansatz} for all $n$, $\eee_n$ converges to zero. Subsequently, we demonstrate that if \eqref{eq:ansatz} is fulfilled
at step $n$, it is also fulfilled for $n+2$. First, we assert that \eqref{eq:ansatz} implies
\begin{equation}
 \eee_n \leq \tel{7}\frac{1-\tau_n^4}{1+\tau_n^4} \leq \tel{7}. \label{eq:EpsBound}
\end{equation}
Now we can upper bound $P_n$,
\begin{equation}\label{eq:PBound}
 P_n \leq \kl{2+5\eee_n^2}\eee_n^2 \leq \frac{11}{5} \eee_n^2.
\end{equation}
Furthermore 
\begin{align}
 \Delta_n &\leq \frac{2Z_nP_n}{1-2Z_nP_n} \leq \frac{2\frac{11}{5}Z_n\eee_n^2}{1-2\frac{11}{5}Z_n\eee_n^2}
 \leq \frac{\frac{2\cdot11}{5\cdot7}\eee_n}{1-\frac{2\cdot11}{5\cdot7}\eee_n} =\frac{\frac{22}{35}\eee_n}{1-\frac{22}{35}\eee_n} \leq \frac{7}{10} \eee_n   \label{eq:DeltaBound} 
\end{align}
is true for all $\eee_n \leq \tel{7}$ and leads us to
\begin{equation}
 \frac{1+\Delta_n}{1-\Delta_n} \leq 1 + \frac{31}{20} \eee_n \label{eq:FracBound}.
\end{equation}
Finally, we can upper bound $\eee_{n+2}$ \je{as}
\begin{equation}
 \eee_{n+2} \leq 4\kl{1 + \frac{31}{20} \eee_n}\kl{1+ \eee_n^2} \leq 5\eee_n^2. \label{eq:Converge}
\end{equation}
This converges to zero for $\eee_0 < \tel{5}$ which is implied by \eqref{eq:ansatz}. Next, we show that if \eqref{eq:ansatz} is true for step $n$ it also holds
for step $n+2$. We look at the update rule of $\tau_n$, \eqref{eq:TauUpdate}, and insert the bounds \eqref{eq:PBound}, \eqref{eq:DeltaBound} and \eqref{eq:FracBound}, \je{to find}
\begin{align}
 \tau_{n+2} &\leq \tau_n^4\kl{1+\frac{21}{10}\eee_n} + \kl{1+ \frac{31}{20}}\kl{\frac{21}{10} + \frac{11}{5}\eee_n} \eee_n \nonumber\\
    &\leq \tau_n^4 + \tau_n^4\frac{21}{10}\eee_n + 3 \eee_n .
\end{align}
Inserting \eqref{eq:EpsBound} yields
\begin{equation} 
 \tau_n^4 + \tau_n^4\frac{3}{10}\frac{1-\tau_n^4}{1+\tau_n^4} + \frac{3}{7} \frac{1-\tau_n^4}{1+\tau_n^4}.
\end{equation}
We \je{investigate} two different cases, $\frac{1}{2}\leq \tau_n \leq 1$ and $0 \leq \tau_n \leq \frac{1}{2}$.
In the first regime, we know
\begin{equation}
 \tau_{n+2} \leq \tau_n^4 + \tau_n^4\frac{3}{10}\frac{1-\tau_n^4}{1+\tau_n^4} + \frac{3}{7} \frac{1-\tau_n^4}{1+\tau_n^4} \leq \tau_n.
\end{equation} 
Therefore, if \eqref{eq:ansatz} is fulfilled, $\tau_{n+2}\leq \tau_n$ and from \eqref{eq:Converge} we know $\eee_{n+2} \leq \eee_n$. If $\eee_n$ and $\tau_n$ satisfy ansatz
\eqref{eq:ansatz}, than $\eee_{n+2}$ and $\tau_{n+2}$ necessarily do the same, since they are both smaller.
In the second regime we know
\begin{equation} 
 \tau_{n+2} \leq \tau_n^4 + \tau_n^4\frac{3}{10}\frac{1-\tau_n^4}{1+\tau_n^4} + \frac{3}{7} \frac{1-\tau_n^4}{1+\tau_n^4} \leq \frac{1}{2},
\end{equation}
so once $\tau_n$ is smaller than $\frac{1}{2}$ it stays smaller than $\frac{1}{2}$. Therefore, 
\begin{equation}
 Z_{n+2}\eee_{n+2} \leq \frac{17}{15}5\eee_n^2 \leq \frac{17}{15}\frac{5}{7} \eee_n \leq \tel{7},
\end{equation}
and the ansatz holds true in step $n+2$. By induction we can conclude that if \eqref{eq:ansatz} is satisfied for $n=0$ it stays satisfied for all $n=2k \with k\in \mathbb{N}$.
The assumption for step $n=0$ can be reformulated in the form
\begin{equation}
 \eee_0 \leq \tel{7} \frac{1-\tau_0^4}{1+\tau_0^4}.
\end{equation}
The speed of the convergence is given by
\begin{equation}
 \eee_n \leq \tel{5}\kl{5\eee_0}^{2^{n/2}}.
\end{equation}
This concludes the derivation of a threshold for the recurrence protocol that ensures convergence to maximally entangled pure states.

\section{Bounding observables}
\label{secPhysCon}

%\subsection{Normalization of an MPO}

So far we always dealt with unnormalised MPOs, since it allows the explicit description of the MPO matrices independent of the length of the MPO. 
An MPO with the same matrices but of a different length will generally have a different norm.
We now \je{turn to showing} that if the transfer matrix is trace-preserving the norm of an MPO will be exponentially close to unity in $L$, the length of the chain. \je{We start by stating a helpful lemma.}

\begin{lm}[\je{Trace bound}]\label{lem:TraceBound}
 If $E$ is isomorphic to $\EE$, a trace-preserving completely positive map, then 
 \begin{equation}
  1 - d^{5/2}\tec{\EE}^L \leq \Trace{E^L} \leq 1+ d^{5/2}\tec{\EE}^L,
  \end{equation}
  where $\tec{\EE}$ is the ergodicity coefficient of $\EE$.
\end{lm}

\begin{proof}
 
The proof is based on the idea to write the trace explicitly in the matrix picture for $E$ as a sum over a basis. We chose the generalised Pauli matrices $P$, which vectorised and normalised
will give us a convenient basis for the trace. We use that except for the identity, all these matrices have trace equal to zero.
%\ec{[What is the trace of a channel!?  Expression doesn't seem well defined.]}\jg{[But $E$ as well as $\EE(\id)$ are matrices.]} %Solved.
\begin{align}
 \Trace{E^{L}}
 &=\Braket{{\id} d^{-1/2}}{E^{L_n-1}_n}{{\id} d^{-1/2}} + \sum_{\sigma\in P/\id}{\Braket{\sigma}{E^{L_n}_n}{\sigma}}.
 \end{align}
 Now we change into the channel picture to use the ergodicity coefficient.
 \begin{align}
 \Trace{E^{L}} &\geq \tel{d}\Trace{\EE\adj\kl{\id}} -  \max_{\Trace{\sigma}=0} d^2\frac{\Trace{\sigma\EEn^{L_n}\kl{\sigma}}}{\norm{\sigma}{2}^2}.
 \end{align}
 Next, we can use that the identity is a fixed point of $\EE\adj$. We can freely lower bound by changing from $\norm{\sigma}{2}$ to $\infNorm{\sigma}$, but we have to introduce a factor ${d}^{1/2}$
 to replace $\norm{\sigma}{2}$ with $\oneNorm{\sigma}$. \je{This gives}
 \begin{align}
 \Trace{E^{L}}
 &\geq 1- d^2\sqrt{d} \max_{\Trace{\sigma}=0} \frac{\Trace{\sigma\EE^{L}\kl{\sigma}}}{\infNorm{\sigma}\oneNorm{\sigma}} \\
 &\geq 1- d^{5/2} \max_{\Trace{\sigma_1}=0} \max_{\Trace{\sigma_2}=0} \frac{\Trace{\sigma_1\EE^{L}\kl{\sigma_2}}}{\infNorm{\sigma_1}\oneNorm{\sigma_2}} \nonumber\\
 &\geq 1- d^{5/2} \max_{\Trace{\sigma_2}=0} \frac{\oneNorm{\EE^{L}\kl{\sigma_2}}}{\oneNorm{\sigma_2}}\nonumber \\
 &\geq 1- d^{5/2} \tec{\EE}^{L}.\nonumber
\end{align}
The upper bound can be achieved analogously.
\end{proof}

%\subsection{Upper Bound of Expectation Values}

We now show that it is possible to upper bound the expectation value of an Hermitian operator with its 1-to-1 norm. The transfer matrix of an operator supported on $r$
sites is denoted as $E_O^r$, and the isomorphic channel as $\EE_O^r$.

\begin{lm}[\je{Bounds to transfer operators}]\label{lem:NormBound}
 If $\EE$ is a trace-preserving completely positive map and $\EE_O^r$ is an operator channel supported on $r$ sites, then
 \begin{equation}
  \frac{\Trace{E_O^rE^{L-r}}}{\Trace{E^L}} \leq \toNorm{E_O^r}{1} \frac{1+d^{5/2}\tec{\EE}^{L-r}}{1-d^{5/2}\tec{\EE}^{L}}
  \end{equation}%\label{lem:NormBound}
\ec{and by setting $r=1$ and $E_{O}=A,B,C,D$ we find the local fidelity w.r.t.\ the respective Bell states.}
\end{lm}

\begin{proof} \je{We start from}
 \begin{align}
 &\Trace{E^r_OE^{L-r}} 
 = \Braket{{\id} d^{-1/2}}{E^r_OE^{L-r}}{{\id} d^{-1/2}} + \sum_{\sigma\in P/\id}{\Braket{\sigma}{E^r_OE^{L-r}_n}{\sigma}} .
\end{align}
As before, we switch to the channel picture to deal with the first term,
\begin{align}
 \Braket{{\id} d^{-1/2}}{E^r_OE^{L_n-r}_n}{{\id} d^{-1/2}} &= \frac{1}{d}\Trace{\EE_O^r\EE^{L-r}\kl{\id}} 
 = \frac{\Trace{\EE_O^r\EE^{L-r}\kl{\id}}}{\Trace{\id}}, \\
 &\leq \toNorm{\EE_O^r\EE^{L-1}}{1},  \nonumber\\
 &\leq \toNorm{\EE_O^r}{1}\toNorm{\EE^{L-1}}{1},  \nonumber\\
 &= \toNorm{\EE_O^r}{1}. \nonumber
\end{align}
The second term will be handled accordingly. In this way, we find
\begin{align}
	\hspace{-1cm}\sum_{\sigma\in P/\id}{\Braket{\sigma}{E_O^rE^{Lr}_n}{\sigma}} 
& \leq d^2 \max_{\Trace{\sigma}=0} {\Braket{\sigma}{E_O^rE^{L-1}_n}{\sigma}} \\ 
 &= d^2 \max_{\Trace{\sigma}=0} \frac{\Trace{\sigma\EE_O^r\EEn^{L_n-1}\kl{\sigma}}}{\norm{\sigma}{2}^2}  \nonumber \\
 &= d^2 \max_{\Trace{\sigma}=0} \frac{\Trace{\sigma\EE_O^r\EEn^{L_n-1}\kl{\sigma}}}{\norm{\sigma}{2}^2}  \nonumber \\
 &\leq d^2 \max_{\Trace{\sigma_1}=0} \max_{\Trace{\sigma_2}=0} \frac{\Trace{\sigma_1\EE_O^r\EEn^{L_n-1}\kl{\sigma_2}}}{\infNorm{\sigma_1}\norm{\sigma_2}{2}}. \nonumber
 \end{align}
 In the last step we made again use of $\infNorm{\sigma}\leq\norm{\sigma_2}{2}$.
 \je{We can hence conclude that}
 \begin{align}
	 \sum_{\sigma\in P/\id}{\Braket{\sigma}{E_O^rE^{Lr}_n}{\sigma}}
& \leq d^2 \max_{\Trace{\sigma_2}=0} \frac{\oneNorm{\EE_O^r\EEn^{L_n-1}\kl{\sigma_2}}}{\norm{\sigma_2}{2}}
\\
 &= d^2\sqrt{d} \max_{\Trace{\sigma_2}=0} \frac{\oneNorm{\EE_O^r\EEn^{L_n-1}\kl{\sigma_2}}}{\oneNorm{\EEn^{L_n-1}\kl{\sigma_2}}}\frac{\oneNorm{\EEn^{L_n-1}\kl{\sigma_2}}}{\oneNorm{\sigma_2}} \nonumber\\
 &= d^{5/2} \toNorm{\EE_O^r}{1}\max_{\Trace{\sigma_2}=0} \frac{\oneNorm{\EEn^{L_n-1}\kl{\sigma_2}}}{\oneNorm{\sigma_2}} \nonumber\\
 &= d^{5/2} \toNorm{\EE_O^r}{1} \tec{\EE}^{L-r}.  \nonumber
\end{align}
For the denominator we use the upper bound from Lemma \ref{lem:TraceBound} and this completes the proof. \end{proof}

Alas, lower bounding in the same fashion is not generally possible without knowing more about the eigenvectors of $E^r_O$. Lemma \ref{lem:NormBound} tells us that upper bounds for physical expectation values can easily be derived by looking at the channel norms, given a long chain and that the transfer matrix is trace-preserving.
This raises the question whether the transfer matrix remains trace-preserving over the course of the iteration.
This is only the case for deterministic protocols without post-selection since no matrices are disregarded and thus the transfer matrix $\EEn$ is a power of the previous transfer matrix.
Consequently, the proof of a threshold was much more straight forward in the deterministic case.
% For this reason, we first derive a threshold that ensures purification in the deterministic setting
% and then explain how we deal with the changing transfer matrix in the postselective case.

\section{Physical model}
\je{The machinery presented above is applicable to large classes of natural preparations
that give rise to correlations resulting from memory effects. We have seen in what way generically
the entanglement distillation protocols can be}
%``renormalised'' in order to arrive at functioning schemes 
\jg{ successfully applied }
\je{in the presence of correlations. In this section, we present a physical model that 
highlights the functioning of the scheme for a particular choice of a Hamiltonian that reflects
a memory effect. This is not meant to be a particularly feasible model -- this will  of course  depend
on the physical architecture, and our formalism is applicable to all those scenarios. It is 
rather meant as a paradigmatic example, to stress the general functioning of the scheme.
}
\label{secPhysHam}
\jg{We start with the following set of initially uncorrelated Werner states} with $F_0$,  
\begin{align}
	\rho_0 &= F_0\phi^+  
	+
	 \frac{1-\je{F_0}}{3}\kl{\phi^-+\psi^++ \psi^-}.
\end{align}
\jg{These states subsequently undergo a unitary interaction with a memory bit on Bob's side,}
$U\kl{t,J}:=\exp\kl{itH}$, for $J>0$,
with
\begin{equation}	
H=J\kl{X\otimes X+Y\otimes Y+Z\otimes Z} + Z\otimes \id + \id\otimes Z.
\end{equation}
We further implement a de-phasing channel for \je{a} forgetful memory, \jg{which is applied to the memory in-between two interactions with Bob's qubit,}
 \begin{equation}
  D_c\kl{\sigma} = \kl{1-c_D}\sigma + c_D\id
 \end{equation}
 \je{for a suitable $c_D\in [0,1]$.}
 This procedure is depicted in \je{Fig.} 5. \je{It should be clear that while this is a
 particular example of a Hamiltonian and dissipative map chosen, this feature }\jg{of having }\je{ an interplay
 between Hamiltonians reflecting interactions with memory and dissipative channels is generic.}
 
 \begin{figure}[htb]
%\tikzstyle{tensor}=[rectangle,rounded corners=3pt,draw=blue!50,fill=blue!20,thick]
%\begin{tikzpicture}
%\node[tensor] (U1) at (1.5,1) {$U(t)$};
%\node[tensor] (U1') at (1.5,-1) {$U'(t)$};
%\node[tensor] (U2) at (5.5,1) {$U(t)$};
%\node[tensor] (U2') at (5.5,-1) {$U'(t)$};
%
%\node (d1) at (0.5,1) {};
%\node (d1') at (0.5,-1) {};
%\node (d2) at (4.5,1) {};
%\node (d2') at (4.5,-1) {};
%
%\draw[-,draw=gray!70] (U1) -- +(0,1) node[left] {$b$} -- (U1') -- +(0,-1) node[left] {$b'$};
%\draw[-,draw=gray!70] (U2) -- +(0,1) node[left] {$b$} -- (U2') -- +(0,-1) node[left] {$b'$};
%\draw[-,draw=gray!70] (d1) -- +(0,1) node[left] {$a$} -- (d1'.south) -- +(0,-0.87) node[left] {$a'$};
%\draw[-,draw=gray!70] (d2) -- +(0,1) node[left] {$a$} -- (d2'.south) -- +(0,-0.87) node[left] {$a'$};
%\draw[-] (U1) -- +(-2,0) node[above] {$m$} -- (U2) -- +(2,0);
%\draw[-] (U1') -- +(-2,0) node[above] {$m'$} -- (U2') -- +(2,0);
%
%\node[tensor] (U1) at (1.5,1) {$U(t)$};
%\node[tensor] (U1') at (1.5,-1) {$U^*(t)$};
%\node[tensor] (U2) at (5.5,1) {$U(t)$};
%\node[tensor] (U2') at (5.5,-1) {$U^*(t)$};
%
%\node[tensor,minimum width = 2cm] (rho1) at (1,0) {$\rho_0$};
%\node[tensor,minimum width = 2cm] (rho2) at (5,0) {$\rho_0$};
%\node[tensor, minimum height = 2.6cm] (D1) at (3,0) {$D_c$};
%\node[tensor, minimum height = 2.6cm] (D2) at (7,0) {$D_c$};
%\end{tikzpicture}
\includegraphics[width=0.53\textwidth]{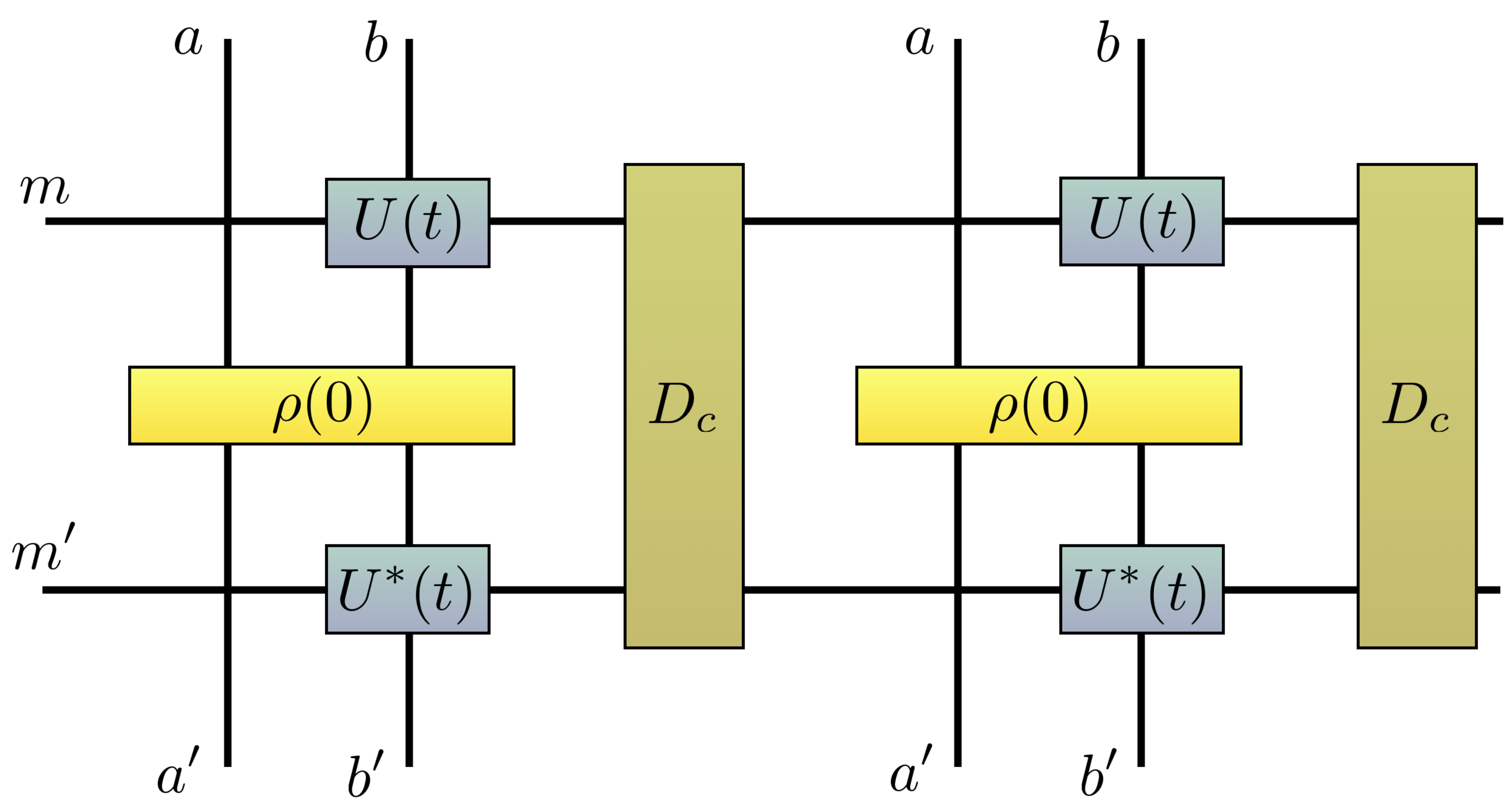}
\label{fig:physHam}
\caption{MPO diagram of the process generating our physical example.}
\end{figure}
We compare the performance of the recurrence protocol on (i) sequentially prepared states with implemented memory with (ii) perfect memoryless i.i.d.\ distributions of the same local fidelity.
As a measure of how a specific memory setting performs we introduce the notion of relative noise
\begin{equation}
  \gamma_n=\frac{1-F_n^{\mathrm{MPO}}}{1-F_n^{\mathrm{i.i.d.}}}
\end{equation}
after $n$ rounds of iteration.

\begin{figure}[ht!]
	\centering
\includegraphics[width=0.55\textwidth]{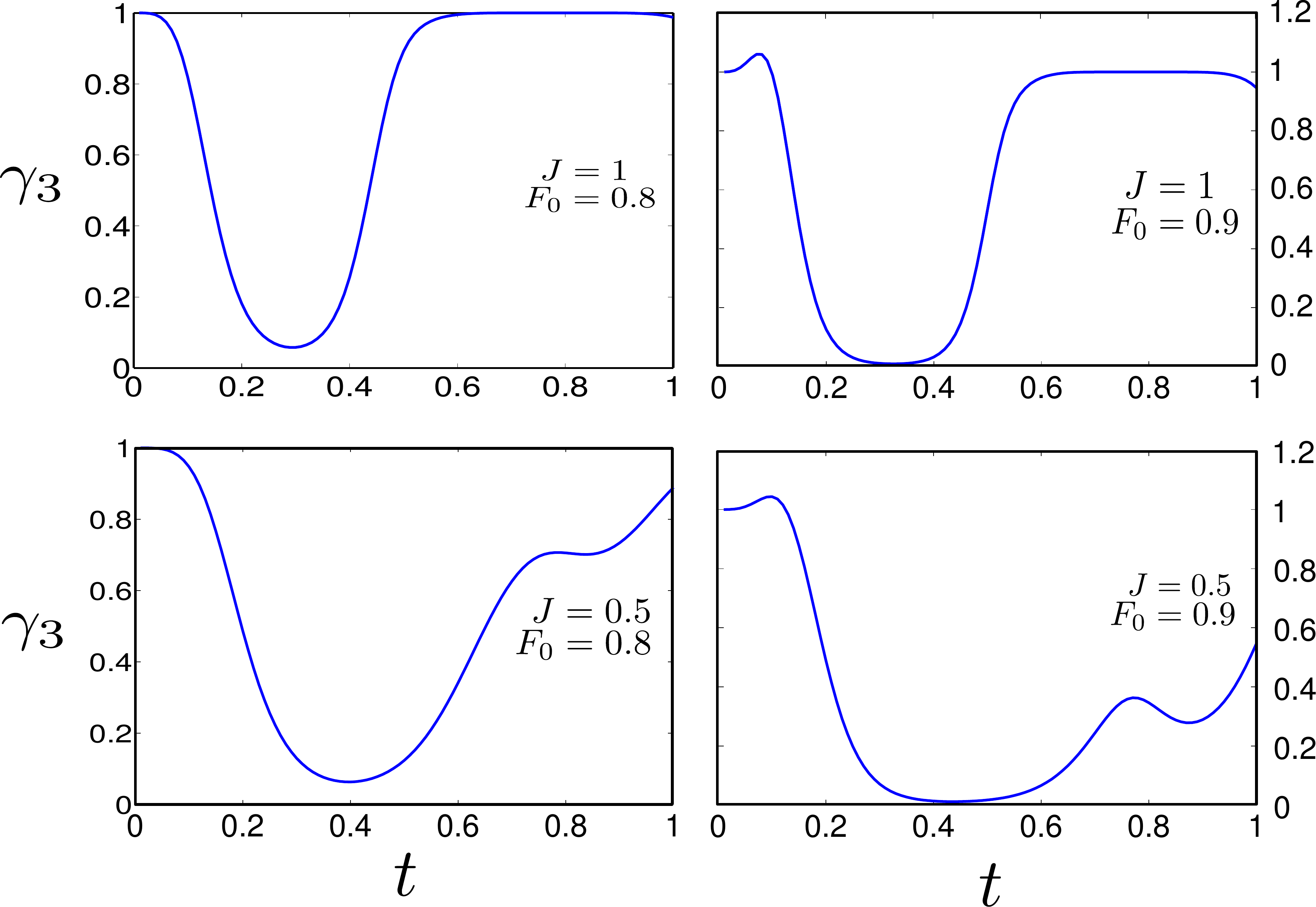}\label{fig:PhysRec}
\caption{Relative noise after three steps of iteration for different initial fidelities $F_0$ and interaction parameter $J$ with $c_D=0.04$. Whenever the relative noise is smaller than one,
we see a comparable advantage over the i.i.d.\ case.}
\end{figure}

For certain parameters, e.g. $F_0=0.9$, $J=1$, $c_D=0.04$ and $t=0.1$, the MPO setting converges significantly faster than the i.i.d.\ setting of the same local fidelity.
After one round we have $\gamma_1 =0.9$, meaning that after one round we have only 90\% of the noise compared to the i.i.d.\ case. We included this specific MPO in Fig.\ \ref{fig:taueps}.
For longer interaction times $t=0.47$ we get a local fidelity of $\leq 0.4$. An i.i.d.\ setting with this fidelity does not succeed, but the MPO setting does.
This shows that we transport the unwanted inter pair correlations introduced by the memory into the wanted correlations between the pairs. We tested our Heisenberg memory model for 
different parameters and the distillation of the correlated states performs better than the distillation in the i.i.d.\ case for a large range of interaction times, see also Fig.\  6.

 \end{document}